\documentclass[a4paper,12pt]{article} 
\usepackage[centertags]{amsmath}
\usepackage{amsmath}
\usepackage[dutch,british]{babel}
\usepackage{amsfonts}
\usepackage{amssymb}
\usepackage{amsthm}
\usepackage{newlfont}
\usepackage{color}
\usepackage{subfig}
 \usepackage{bbm}

\usepackage{stmaryrd}
\usepackage{mathrsfs}
\usepackage{pstricks,pst-node}
\usepackage{palatino}  
\usepackage{fancyhdr}
\usepackage{graphicx}
\usepackage{fancybox}
\usepackage{multibox}
\usepackage{geometry}
 \usepackage{psfrag}
 
 \usepackage{multirow}

\theoremstyle{plain}
  \newtheorem{theorem}{Theorem}[section]
  
  \newtheorem{proposition}[theorem]{Proposition}
  \newtheorem{lemma}[theorem]{Lemma}
  
\theoremstyle{definition}
  
  \newtheorem{assumption}[theorem]{Assumption}

\theoremstyle{remark}

\numberwithin{equation}{section}

\DeclareMathOperator{\Tr}{Tr}

 \DeclareMathOperator{\supp}{Supp}
\renewcommand{\Re}{\mathrm{Re}\, }

\newcommand\otimesal{\mathop{\hbox{\raise 1.6 ex
  \hbox{$\scriptscriptstyle\mathrm{al}$}
\kern -0.92 em \hbox{$\otimes$}}}}
\newcommand\oplusal{\mathop{\hbox{\raise 1.6 ex
  \hbox{$\scriptscriptstyle\mathrm{al}$}
\kern -0.92 em \hbox{$\oplus$}}}}
\newcommand\Gammal{\hbox{\raise 1.7 ex
\hbox{$\scriptscriptstyle\mathrm{al}$}\kern -0.50 em $\Gamma$}}
\renewcommand\i{\mathrm{i}}

\let\al=\alpha   \let\ep=\epsilon

  \let\ga=\gamma 
\let\ka=\kappa \let\la=\lambda \let\om=\omega

 \let\Ga=\Gamma \let\La=\Lambda \let\Om=\Omega

\newcommand{\caA}{{\mathcal A}}

\newcommand{\caE}{{\mathcal E}}

\newcommand{\caG}{{\mathcal G}}

\newcommand{\caI}{{\mathcal I}}

\newcommand{\caT}{{\mathcal T}}

\newcommand{\scrB}{{\mathscr B}}

\newcommand{\scrH}{{\mathscr H}}

\newcommand{\scrR}{{\mathscr R}}
\newcommand{\scrS}{{\mathscr S}}
\newcommand{\scrT}{{\mathscr T}}

\newcommand{\bbC}{{\mathbb C}}

\newcommand{\bbE}{{\mathbb E}}

\newcommand{\bbN}{{\mathbb N}}

\newcommand{\bbR}{{\mathbb R}}

\newcommand{\bbT}{{\mathbb T}}

\newcommand{\bbZ}{{\mathbb Z}}

\newcommand{\opunit}{\text{1}\kern-0.22em\text{l}}

\newcommand{\frt}{{\mathfrak t}}

\newcommand{\e}{{\mathrm e}}

\renewcommand{\d}{{\mathrm d}}
\newcommand{\sys}{{\mathrm S}}
\newcommand{\res}{{\mathrm R}}

\renewcommand{\sp}{\mathrm{sp}}

\newcommand{\Dom}{\mathrm{Dom}}

\newcommand{\beq}{ \begin{equation} }
\newcommand{\eeq}{ \end{equation} }
\newcommand{\bet}{ \begin{theorem} }
\newcommand{\eet}{ \end{theorem} }

\newcommand{\baq}{\begin{eqnarray}}
\newcommand{\eaq}{\end{eqnarray}}

\renewcommand{\supp}{\mathrm{Supp}}

\newcounter{qcounter}

 \newcounter{smallarabics}
\newenvironment{arabicenumerate}
{\begin{list}{{\normalfont\textrm{\arabic{smallarabics})}}}
  {\usecounter{smallarabics}\setlength{\itemindent}{0cm}
  \setlength{\leftmargin}{5ex}\setlength{\labelwidth}{4ex}
  \setlength{\topsep}{0.75\parsep}\setlength{\partopsep}{0ex}
   \setlength{\itemsep}{0ex}}}
{\end{list}}

\newcounter{smallroman}

\newcommand{\ben}{\begin{arabicenumerate}}
\newcommand{\een}{\end{arabicenumerate}}

\newcommand{\norm}{ \|}
\newcommand{\str}{ |}

\newcommand{\lat}{ \bbZ^d }
\newcommand{\tor}{ {\bbT^d}  }

\newcommand{\initialres}{\rho_\res^{\mathrm{ref}}}
\newcommand{\initialresfinite}{\rho_\res^{\mathrm{ref}, \La}}

\newcommand{\links}{\mathrm{Le}}
\newcommand{\rechts}{\mathrm{Ri}}
\newcommand{\adjoint}{\mathrm{ad}}
\newcommand{\ad}{\adjoint}

\newcommand{\dist}{\mathrm{dist}}
\newcommand{ \hook}{\mathrm{Hook}}
\newcommand{ \dt}{\tau}
\newcommand{\weird}{\#}
\newcommand{\poly}{\mathrm{Pol}}
\newcommand{\inter}{\mathrm{Int}}

\newcommand{\uw}{\underline{w}}
\newcommand{\ut}{\underline{t}}

\begin{document}

\begin{center}
\large{ \bf{Return to equilibrium' for weakly coupled quantum systems: a simple polymer expansion}}
 \\
\vspace{15pt} \normalsize

{\bf   W.  De Roeck\footnote{
email: {\tt
 w.deroeck@thphys.uni-heidelberg.de}}  }\\
\vspace{10pt} 
{\it   Institut f\"ur Theoretische Physik  \\ Universit\"at Heidelberg \\
Philosophenweg 19,  \\
D69120 Heidelberg,  Germany 
} \\

\vspace{15pt}

{\bf   A. Kupiainen\footnote{
email: {\tt    antti.kupiainen@helsinki.fi  }}  }\\
\vspace{10pt} 
{\it   Department of Mathematics \\
 University of Helsinki \\ 
P.O. Box 68, FIN-00014,  Finland 
} \\

\end{center}

\vspace{20pt} \footnotesize \noindent {\bf Abstract: }
Recently, several authors studied small quantum systems weakly coupled to free boson or fermion fields at positive temperature. All the rigorous approaches we are aware of employ complex deformations of Liouvillians or Mourre theory (the infinitesimal version of the former).  We present an approach based on polymer expansions of statistical mechanics.  Despite the fact that our approach is elementary, our results are  slightly sharper than those contained in the literature up to now.    We show that, whenever the small quantum system is known to admit a Markov approximation (Pauli master equation \emph{aka} Lindblad equation) in the weak coupling limit, and the Markov approximation is exponentially mixing,  then the  weakly coupled system approaches a unique invariant state that is perturbatively close to its Markov approximation.

%

\section{Introduction}\label{sec: introduction}

\subsection{Motivation}
Quantum systems consisting of a small subsystem (say, an atom) and a large component (say, a heat bath) have received a lot of attention lately, sparked by the elegant results of \cite{lindenpopescushortwinter, lebowitzgoldsteinmastrodonatotumulka, reimannfoundations}.   The challenge in this problem is to prove that the subsystem thermalizes under influence of the heat bath, this property will be called 'return to equilibrium', or simply  'RTE' hereafter. The quoted works show that for typical such systems (in a precisely defined sense of typicality) the subsystem is close to equilibrium for most times. However, if one wants to study in
more detail the subsystem dynamics one needs to resort to concrete models. Indeed, such subsystem-reservoir models have been successfully and rigorously studied since the late $90$'s (we refer to \cite{jaksicpillet2,bachfrohlichreturn} for early results in this field) under the assumptions
\ben
\item that the heat bath consists of a free (and hence explicitly solvable) field. 
\item that the coupling between subsystem and heat bath is small compared to the energy scales of the subsystem
\een
Nevertheless, these results either impose rather strong assumptions on the form of the system-reservoir coupling, or  they become quite involved technically. 
The aim of the present paper is to develop an intuitive and simple approach for RTE in this case.   To explain our result, let us first recall that these systems were already studied in the 70's from the point of view of quantum master equations; B.\ Davies  \cite{davies1} pioneered the rigorous derivation of master equations in this framework, thus making precise earlier heuristic ideas of I.\ Prigogine and P.\ Van Hove.  The master equation is derived, under mild conditions on the form of the coupling, by scaling time $t$ as $t \sim \la^{-2}$ ($\la$ is the coupling strength) and taking $\la \to 0$. It exhibits all irreversible phenomena expected in such model systems and as such, it has inspired many researchers in open quantum systems.   However, it  does of course \textbf{not} yield information on the long time (longer than $\la^{-2}$) behaviour of the system, 

We prove  that, \textbf{if} the condition necessary for the derivation of the master equation is satisfied and the master equation is exponentially ergodic (exhibits exponentially fast return to equilibrium), \textbf{then} the system thermalizes in the long time limit, for small but nonzero coupling strength $\la$ (more generally, it reaches a steady state, since we do not assume that the heat bath(s) is(are) in equilibrium). Moreover, we give an explicit bound on the speed of convergence towards the steady state.  The necessary condition is the integrability in time of certain correlation functions of the free field. 
To our best knowledge, this condition is weaker than that of other RTE- results in the literature. 

\vskip 3mm
\noindent {\bf Acknowledgements}.  This work was done while $W.D.R$ was a post-doc at the university of Helsinki. We thank the European Research Council
and the Academy of Finland  for financial support.  W.D.R. is grateful to Kevin Schnelli for serious proofreading and for pointing out several errors in the manuscript.

\subsection{Setup}

Let $\scrH_\sys$ be a finite-dimensional Hilbert space (modeling the small system) with a Hamiltonian $H_\sys$ (a Hermitian matrix).  To describe the field that plays the role of reservoir, we first pick a finite, discrete hypercube  $\La=\La_L =  \bbZ^d \cap\,  ( -L, L]^d $  for some $L \in \bbN$ and we enclose the field in the volume $\La$ (we could as well choose $\La$ to be a box in $\bbR^d$).  Since we will mainly use the Fourier transform, we define the set of (quasi-)momenta $\La^* = (\pi\bbZ/L)^d \cap \, (-\pi, \pi]^d$.  The dynamics of one reservoir excitation is given by the one-particle dispersion relation $\omega^{\La}(q)$ and the Hamiltonian of the whole field in $\La$ is given by 
\beq
H^{\La}_\res:= \sum_{q \in \La^*} \omega^{\La}(q) a^*_q a_q
\eeq
acting on the bosonic (symmetric) Fock space $\scrH_\res^\La = \Ga(l^2(\La^*))$. 
Here $a_q^*, a_q$ are the creation/annihilation operators of a mode with (quasi-)\-momentum $q \in \La^*$, satisfying the canonical commutation relations $[a_q, a^*_{q'}] = \delta_{q,q'}$. 
The Hilbert space of the total system consisting of small system and field, is  $\scrH^{\La}= \scrH_\sys \otimes \scrH^{\La}_\res$, and we simply write  $H_\sys$ and  $H_\res^{\La}$ for the operators $H_{\sys} \otimes 1$ and $1 \otimes H^{\La}_\res$ acting on $\scrH^\La$.
 The coupling between field and the small system  
  is assumed to be linear in the creation and annihilation operators and it can hence be written in the form 
 \beq
H_{\inter}^{\La}:= \sum_{i \in \caI}  D_i \otimes \Phi(\phi^{\La}_i), \qquad    \Phi(\phi^{\La}_i):=    \sum_{q \in \La^*} \left(\phi^{\La}_i(q) a^*_q +\overline{\phi_i^{\La}(q)} a_q \right)   \label{def: interaction ham}
\eeq
where $D_i=D_i^{*}$ are self-adjoint operators on $\scrH_\sys$ and 
$\phi_i^{\La}$ are  functions (form factors) to be specified.  $\caI$ is a finite index set. 

The total Hamiltonian of the system is hence, with a coupling strength $\la \in \bbR$, 
\beq
H^{\La}_{\la}:= H_\sys +  H^{\La}_\res+  \la H_{\inter}^{\La}, \qquad  \textrm{on} \,\, \scrH^{\La}.
\eeq
A standard application of the Kato-Rellich theorem states that, if $(\om^{\La})^{-1/2} \phi_i^{\La} \in l^2(\La^*)$, then  $H^{\La}$ is self-adjoint on the domain of $H^{\La}_\res$. 

Initially, the field is in a Gaussian state that will be called the 'reference state':
The density matrix $\initialresfinite$ of this reference state is specified by the correlation functions of smeared field operators, i.e.\ by 
\beq  \label{def: correlation functions finite volume}
\Tr_\res \left[ \initialresfinite \Phi(\varphi_1)  \ldots \Phi(\varphi_n) \right], \qquad   \eeq
where $\varphi_i$ are functions on $\La^*$ and $\Tr_\res$ is the trace on $\scrH^{\La}_\res$.

Our main assumptions on this reference state  $\initialresfinite$ are that it is
 \ben
 \item stationary w.r.t.\  the decoupled dynamics, i.e. \beq
 \label{eq: stationarity of reservoir}
     \e^{-\i t H^{\La}_\res }  \initialresfinite  \e^{\i t H^{\La}_\res } = \initialresfinite
 \eeq
 \item gauge-invariant, i.e.\ all correlation functions \eqref{def: correlation functions finite volume} that involve an odd number of field operators, are zero.
 \item Gaussian (also called "quasifree"), i.e.\  the  higher correlation functions are related to the two-particle correlation function via the Gaussian relation
 \beq \label{eq: gaussian property}
\Tr_\res\left[\initialresfinite \Phi(\varphi_1) \ldots \Phi(\varphi_n) \right]  =   \sum_{\mathrm{pairings}\,  \pi} \prod_{(i,j) \in \pi} \Tr_\res \left[\initialresfinite  \Phi(\varphi_i)  \Phi(\varphi_j) \right]
\eeq
where a 'pairing' $\pi$ is a set of pairs $(i,j)$ with the convention $i<j$.
 \een
 By standard theory of Gaussian (or quasi-free) states, the  above properties imply  that the state $\initialresfinite$ is completely determined by a positive density function $0<\eta^{\La}(q) < \infty$ via the relations
  \beq
\Tr_\res \left[\initialresfinite  \Phi(\varphi)
\Phi(\varphi') \right] =  \langle \varphi', \eta^{\La} \varphi \rangle+  \langle \varphi, (1+  \eta^{\La} ) \varphi'\rangle  \label{eq: examples correlation functions}
\eeq
The invariance of $\initialresfinite$ under the free field dynamics is ensured by the commutation relation $[\om^\La,\eta^{\La} ]=0$.
Since the field is a finite collection of harmonic oscillators, the reference state $\initialresfinite$ is a well-defined traceclass density matrix. 
The material in this section is completely standard and we refer the reader to e.g.\ \cite{derezinski1, bratellirobinson} for details that were omitted here (note however that these texts deal with infinite volume $\La$ from the start and hence they are  necessarily more involved technically).

\subsection{Thermodynamic limit}

As long as $\La$ is finite, we cannot expect the system to have good ergodic properties and hence we will perform the thermodynamic limit as a first step. 
By  the thermodynamic limit $\La \nearrow \bbZ^d$, we mean that $L \nearrow \infty$, hence the volume $\La$  tends to $\bbZ^d$ and the set $\La^*$ tends to $\bbT^d$.
As we will see, the influence of the reservoir on the dynamics of the subsystem can be expressed entirely in terms of the correlation functions
\baq
f^{\La}_{i,j}(t)& := & \Tr_\res [\initialresfinite \Phi( \e^{\i t \om^{\La}}\phi^{\La}_i)  \Phi(\phi^{\La}_j)] \\
&=&   \langle \phi_j^{\La}, \eta^{\La} \e^{\i t \om^{\La}}  \phi_i^{\La} \rangle +  \langle \phi_i^{\La}, (1+\eta^{\La}) \e^{-\i t \om^{\La}}  \phi_j^{\La} \rangle     
\eaq
Note that $\overline{f^{\La}_{i,j}(t)}= f^{\La}_{j,i}(-t)$ by stationarity of the state $\initialresfinite$.
To discuss the thermodynamic limit of the small system behaviour, it suffices to ask that the $f^{\La}_{i,j}(t)$ converge
\beq
 f_{i,j}(t)= \lim_{\La \nearrow \bbZ^d}  f^{\La}_{i,j}(t), \qquad     \label{eq: thermo correlation}
\eeq
 uniformly in $t$ on any compact set, and that $\sup_t   \str f_{i,j}(t)  \str < \infty $.
We recall that a density matrix on $\scrH_\sys$ is a positive traceclass operator, i.e., belonging to $\scrB_1(\scrH_\sys)$, whose trace is $1$ (of course, since  $\scrH_\sys$ is finite-dimensional, any operator is traceclass). In what follows, we let $\Tr_\res$ stand for the partial trace over the field degrees of freedom, mapping density matrices on $\scrH^{\La}$ into density matrices on $\scrH_\sys$.
 \begin{lemma} \label{lem: thermodynamic limit}
 Assume that  the $f^{\La}_{i,j}(t)$ converge to bounded functions $ f_{i,j}(t)$, uniformly on compacts (i.e.\  \eqref{eq: thermo correlation}).
Then, the thermodynamic limit
\beq
 \rho_{\sys, t}  :=  \lim_{\La \nearrow \bbZ^d}   \Tr_{\res}    \left[ \e^{-\i t H^{\La}_\la} \left( \rho_{\sys, 0} \otimes \initialresfinite  \right)  \e^{\i t H^{\La}_\la}  \right]
\eeq
exists for any initial density matrix $\rho_{\sys, 0}$ on $\scrH_\sys$.
\end{lemma}
The proof of Lemma \ref{lem: thermodynamic limit} is given in Section \ref{sec: combinatorics}. 
 
\subsection{Markov approximation} \label{sec: markov approximation}

The Markov approximation to the model introduced above amounts to replacing the correlation function $f_{i,j}(t)$ by  a multiple of $\delta(t)$ (no memory).  It can be justified in the weak coupling scaling limit $\la \to 0$, if one rescales time as $t \rightarrow \la^{-2}t$. We state this important result precisely in Section \ref{sec: dissipativity}.  For now, we just introduce the precise form of the Markov approximation since one of our assumptions refers to it. 
First, we introduce the \emph{left} and \emph{right multiplication operators} $M_{\links}(A), M_{\rechts}(A)$;
\beq
M_{\links}(A) S :=  A S, \qquad     M_{\rechts}(A)S :=SA^*, \qquad    A, S \in \scrB(\scrH_\sys)  \label{eq: multiplication operators}
\eeq 
Then we set 
\beq
 \tilde L := \sum_{\scriptsize{\left.\begin{array}{c} k_1,k_2 \in \{\links, \rechts \}  \\    i,j \in \caI \end{array}\right.}   }     \, \int_0^{\infty} \d t    \, \e^{\i t \adjoint(H_\sys)} \,   M_{k_2}(\i D_j)   \, \e^{-\i t \adjoint(H_\sys)} \, M_{k_1}(\i D_i)  \left\{ \begin{array}{ccc}        f_{j,i}(t)        &\textrm{if} &    k_{1}=\links \\[2mm]
 \overline{f_{j,i}(t)}   &  \textrm{if} &      k_{1}=\rechts  
 \end{array} \right. \label{def: lindblad}
\eeq
where $\adjoint(H_\sys)=[H_\sys,\cdot] $ and  the integral over $t$ is well-defined by  integrability of $f_{i,j}(\cdot)$, which will be assumed below. Finally
\beq
L :=  \mathop{\lim}\limits_{t \to \infty}\frac{1}{t} \int_{0}^t \d s \,  \e^{\i s \ad(H_\sys)}  \tilde L  \e^{-\i s \ad(H_\sys)}  \label{eq: spectral averaging}
\eeq
where  the limit $t \to \infty$ exists  since $H_\sys$ has discrete spectrum. Note also that $L$ commutes with $\adjoint(H_\sys)$ as follows from the spectral averaging in \eqref{eq: spectral averaging}.  As is discussed in many places, the Lindblad operator $L$ generates a contractive  semigroup $\e^{\frt L}, \frt \geq 0$ on $\scrB_1(\scrH_\sys)$ that is trace-preserving and positivity preserving. In other words, $\e^{\frt L}$ maps the set of density matrices on $\scrH_\sys$ into itself.  
In the above formulas, we denote time by the gothic symbol $\frt$ to emphasize that it corresponds physically to a rescaled time. Indeed, the Lindblad operator $L$ describes the dynamics on long time scales, see Section \ref{sec: dissipativity}. 
 Lindblad operators were first introduced in \cite{lindblad}, an excellent exposition on the properties of $L$ and its derivation from microscopic models can be found  in \cite{lebowitzspohn1}.

\subsection{Result}

We need an assumption on the decay of temporal correlations of  the `free reservoir correlation functions'.
\begin{assumption}[Decay of correlations] \label{ass: decay of correlations}   Recall the correlation functions $f_{i,j}$ introduced in \eqref{eq: thermo correlation}. We assume that
\beq   \label{def: function h}
   \int_0^{\infty} \d t \,    h(t)   < \infty, \qquad  \textrm{where} \quad   h(t):= \sum_{i,j \in \caI} \norm D_i \norm \norm D_j \norm  \str f_{i,j}(t)  \str
\eeq
\end{assumption}

The second assumption concerns  the Lindblad generator $L$, defined in Section \ref{sec: markov approximation}.
\begin{assumption}[Fermi Golden Rule] \label{ass: fermi golden rule} 
The operator $L$ has a simple eigenvalue at $0$. All other eigenvalues lie in the region $\{ z \in \bbC\, \big\str\,  \Re z < -\mathrm{gap}_L \} $ for some $\mathrm{gap}_L>0$. 
\end{assumption}
Obviously, Assumption \ref{ass: fermi golden rule}  and the fact that $\e^{\frt L}$ preserves density matrices, imply  that there is a unique density matrix, $\rho_\sys^{L}$,  such that $L\rho_\sys^{L} =0$  and 
\beq
 \norm \e^{ \frt L }- \str \rho_\sys^{L} \rangle \langle 1 \str \,  \norm   \leq C_L \e^{-\mathrm{gap}_L \frt},  \qquad  \textrm{for all $\frt>0$ and some }   C_L< \infty 
\eeq
where  $\norm \cdot \norm$ is the operator norm of operators acting on $\scrB_1(\scrH_\sys)$ and we use the notation $\str A \rangle \langle A' \str$ to denote the rank-$1$ operator that acts as $S \to  (\Tr[(A')^* S]) A $ with $S,A \in \scrB_1(\scrH_\sys)$ and $A' \in  \scrB(\scrH_\sys) $.
For us it is more convenient to define a characteristic time $\frt_L > 1/ \mathrm{gap}_L$ such that 
\beq
\norm \e^{\frt  L }- \str \rho_\sys^{L} \rangle \langle 1 \str \,  \norm  \leq \e^{- \frt/ \frt_L}, \qquad  \textrm{for} \, \frt> \frt_L
\label{eq: decay scale}
\eeq

Conditions that imply  Assumption \ref{ass: fermi golden rule}  have been discussed extensively, see e.g.\ \cite{frigerio,spohnapproach}.  
Here, we prefer to give a (rather generic) example where the Assumption \ref{ass: fermi golden rule}  can be checked very explicitly:  Assume that the Hamiltonian $H_\sys$ is non-degenerate, hence its spectral projections,
$P(e), e \in \sp H_\sys$, are one-dimensional. 
Then  Assumption \ref{ass: fermi golden rule}  is satisfied if and only if the continuous-time Markov process\footnote{Since $L$ commutes with $\adjoint(H_\sys)$ and preserves positive density matrices, it sends the set of density matrices diagonal in $H_\sys$-basis into itself. Since these diagonal density matrices can be identified with probability measures on $\sp H_\sys$, $\e^{\frt L}$ determines a Markov process on $\sp H_\sys$, namely the one defined by the rates \eqref{def: rates markov process}. } with (finite) state space $\sp H_\sys$ and jump rates
\beq
rate(e \rightarrow e') = \sum_{i,j}    \hat f_{i,j}(e-e')\Tr [P(e') D^*_j P(e) D_i], \qquad \textrm{where} \quad  \hat f_{i,j}(\om) = \frac{1}{2\pi} \int_{\bbR} \d t  \,  \e^{-\i  t \om } f_{i,j}(t)  \label{def: rates markov process}
\eeq
is ergodic.  This in turn can be checked by the Perron-Frobenius theorem: a sufficient condition for ergodicity is that for any two eigenvalues $e,e'$, there is a path $e_0,e_1, \ldots, e_n$ with $e_0=e,e_n=e'$ such that,  for all $i$,  $rate(e_i \rightarrow e_{i+1}) \neq 0$.
We are now ready to state our main result

\bet \label{thm: main}
Assume that Assumption \ref{ass: decay of correlations}  and Assumption \ref{ass: fermi golden rule}  hold and let  $\rho_{\sys, t}$ be defined as in Lemma \ref{lem: thermodynamic limit}. Then,
there is a $\la_0>0$ such that for all $\la$ satisfying $0 < \str \la \str < \la_0$, we have 
\beq
  \lim_{t \to \infty}  \rho_{\sys, t}   =  \rho_{\sys}^{inv}
\eeq
where  the ``invariant density matrix" $ \rho_{\sys}^{inv} = \rho_{\sys}^{inv}(\la) $ does not depend on the initial state  $\rho_{\sys,0}$.  Moreover, $\rho_{\sys}^{inv}$ is a small perturbation of $\rho_{\sys}^L$, the invariant density matrix predicted by the Markov approximation; \beq
\norm   \rho_{\sys}^{inv} -  \rho_{\sys}^L \norm  \rightarrow 0 \quad \textrm{as}  \quad \la   \to 0   \label{eq: thm statement of convergence}
\eeq
\eet
In  \eqref{eq: thm statement of convergence},  $\norm \cdot \norm$ is the operator norm of operators acting on $\scrB_1(\scrH_\sys)$ (although it does not matter since $\scrH_\sys$ is finite-dimensional).

To quantify  the speed of convergence towards the steady state  $ \rho_{\sys}^{inv} $, we need to know the decay properties of the function $h(\cdot)$ that was introduced in Assumption \ref{ass: decay of correlations}.   Let $\zeta(\cdot)$ be a nondecreasing  function on $\bbR^+$ satisfying the conditions
\beq
1 \leq \zeta(t+t')  \leq      \zeta(t) \zeta(t'), \qquad \textrm{for any} \,  t,t' \in \bbR^+    \label{eq: subadditivity}
\eeq
We assume that this function governs the decay of the bath  correlation function $h$, in the sense that
\beq
\int_0^{\infty} \d t h(t)  \zeta(t) < \infty   \label{eq: decay of h in zeta}
\eeq
The case where $ \zeta(t)$ can be chosen to be exponentially increasing, is particularly simple but introduces a complication to the statement of the following result. Therefore we exclude this case explicitly by demanding 
\beq 
\int_0^{\infty} \d t \,  \e^{-\ka t}\zeta(t) <\infty, \qquad  \textrm{for any} \,\, \ka>0   \label{eq: zeta subexponential}
\eeq

\begin{proposition} \label{prop: speed of convergence} Assume the conditions of Theorem \ref{thm: main}.
Let $\zeta$ be a non-decreasing function as above, satisfying (\ref{eq: subadditivity}- \ref{eq: decay of h in zeta}-\ref{eq: zeta subexponential}), and let $\frt_L$ be chosen such that \eqref{eq: decay scale} holds.
Then,  for $\str \la \str$ small enough, 
\beq
\norm \rho_{\sys, t}   - \rho_{\sys}^{inv} \norm  \leq    \exp{\left(- \frac{ \la^{2}t}{ {\frt_L} + o(\str \la \str^0) }  \right) } + o(\str\la\str^0) \left(\zeta \left(\frac{\la^{2}t }{ 2{\frt_L}}\right)\right)^{-1}, \qquad \textrm{for any} \, \,  t >  {\frt_L} \la^{-2}   \label{eq: prop speed}
\eeq
\end{proposition}

Note that  Proposition \ref{prop: speed of convergence} makes no claim about the reduced dynamics $\rho_{\sys,t}$ for short times $ t <  \la^{-2}  {\frt_L}$. The restriction to long times is natural since, for times shorter than $ \la^{-2}  {\frt_L}$,  the exponential decay of the semigroup is not yet visible. For those times,  $\rho_{\sys, t} $ is however well-described by the Markov approximation, see Theorem \ref{thm: weak coupling}.  
On the RHS of \eqref{eq: prop speed}, the time $t$ appears essentially in the combination $ \la^2 t/\frt_L$.  As far as the first term is concerned, this is natural since that term originates from the Markov approximation, i.e.\ the temporal decay embodied in that term takes place on the macroscopic time scale $\sim \la^{-2}\frt_L$.  The second term, however, comes from the slow decay of the reservoir correlation function $h(t)$ on the microscopic time scale, and as such it is not clear why that decay gets prolonged to the macroscopic scale in \eqref{eq: prop speed}. The estimate in that second term is indeed far from optimal (note also the weird factor $'2'$ multiplying $\frt_L$) and this is due to the generality of our result. If, for example, one assumes that $\zeta(t) \sim \str t \str^{\al}, \al>0 $, then one can  state a sharper and more explicit bound.

\subsection{ Discussion and comparison with earlier results}\label{sec: discussion}

\subsubsection{Restriction to confined systems} \label{eq: restriction to confined}
Our result is suited for confined small systems. We explain this in more detail and we distinguish essential assumptions from those made for convenience. 
\begin{list}{ \Alph{qcounter}:~}{\usecounter{qcounter}}
\item The assumption that the 'atom' Hilbert space $\scrH_\sys$ is finite-dimensional, seems not crucial to us. Atoms with an infinite number of energy levels (like the harmonic oscillator) should be treatable with the same technique. A complication that does arise in such infinite-dimensional atoms is that the relaxation of the Markov semigroup is in general not exponential since, in the absence of very energetic field quanta, the atom needs a large time to cascade from a very energetic level to the low-lying levels. We believe however  that this can be remedied by a change of norm on (a subspace of) $\scrH_\sys$ that renders the relaxation exponential, at least for a certain class of interaction Hamiltonians.

\item  The restriction to atom-bath couplings that are linear in the field operators is for notational simplicity only. One can study quadratic coupling in the same way.  Coupling terms of higher order do not yield a well-defined Hamiltonian for bosonic baths, although they are well-defined for fermionic systems. In that case (fermionic baths with coupling of order at least $3$) one has to use sign cancellations to control the Dyson expansion (this is done e.g.\ in \cite{jaksicpautratpillet}) and in such cases an operator-theoretic treatment might be favorable.  
\item  The real assumption that excludes application of our result to extended systems is Assumption \ref{ass: decay of correlations} and more concretely, the sum over $i,j \in \caI$.  For an extended system, the simplest translation invariant coupling would be of the form
\baq
H_{\inter}  &= &    \int_{\tor} \d q \varphi(q)  \e^{\i q X} a_q + \overline\varphi(q)  \e^{-\i q X} a^*_q    \label{eq: trans invariant ham}  \\
&=&   \sum_{x  \in \lat}      \str x \rangle \langle x \str \otimes \Phi (\varphi_x), \qquad \textrm{with}\,  \,   \varphi_x(q) = \e^{\i q x} \varphi(q)
\eaq
where we have taken $\La=\bbZ^d$. 
 The expression on the second line is of the form \eqref{def: interaction ham} with the index set $\caI =\lat$. 
 Even though one could demand that the correlation functions are integrable in time in the sense that 
  \beq
\sup_{x,x'} \int_0^{+\infty} \d t \,  h_{x,x'}(t)   <\infty,
 \eeq
 with
 \beq
 h_{x,x'}(t)  :=  \lim_{\La \nearrow \bbZ^d} \left\str \Tr_\res\left[ \initialresfinite   \Phi ( \e^{\i t \om}\varphi_x)   \Phi (\varphi_{x'})   \right]   \right\str,  
 \eeq
 then still Assumption \ref{ass: decay of correlations} cannot hold because of the sum over $x,x' \in \lat$. In fact, the appearance of the double sum is artificial and one can arrange to have a single sum, and moreover, $h$ depends on the difference $x'-x$ only. Hence,  Assumption \ref{ass: decay of correlations}  would boil down to 
  \beq
   \int_0^{\infty} \d t  \sum_{x} h_{0,x}(t)     < \infty   \label{eq: condition for extended}
 \eeq
 and this assumption cannot be satisfied for any interaction Hamiltonian of the type \eqref{eq: trans invariant ham}. 

\end{list}

\subsubsection{Interacting reservoirs} \label{sec: interacting reservoirs}

Models where the heat bath is not free, i.e.\ it is made up of a genuinely interacting system, are a far dream at this moment. However, we would like to draw attention to the fact that, in contrast to earlier results, our method does not exclude such reservoirs \emph{per se}.  Indeed, the important ingredient of our analysis is a temporal decay condition on the reservoir correlation function. This condition is stated in Assumption \eqref{ass: integrability}, and, for free reservoirs, it is satisfied provided that Assumption \ref{ass: decay of correlations} holds.   The huge challenge is of course to prove such a condition for an interacting system. First steps in this direction have recently been taken in \cite{lukkarinenspohnquantumfluids}.

\subsubsection{Algebraic quantum dynamical systems}
In the literature on the subject, mixing properties are mostly investigated in a more general framework, allowing for initial states that are not factorized (but still local perturbations of $\initialres$) and treating observables that depend on the field as well (since Theorem \ref{thm: main} deals with the reduced dynamics, we get information on observables of the small system only).  In particular, one usually studies the system in the framework of $C^*$ or $W^*$-algebras, in which the concepts  ``ergodicity" and ``mixing" have a natural meaning, inherited from the theory of dynamical systems.  For an introduction to these matters, we refer to \cite{derezinski1,bratellirobinson}. 
It is  straightforward to extend our approach such as to  prove mixing in the above sense,  but since this asks for more notation in Section \ref{sec: discretization}, we have opted not to do so.  The same remark applies to the study of multitime-correlation functions of   small system observables. Our technique shows that these correlation functions are perturbatively close to correlation functions calculated within the Markovian approximation\footnote{Yet, they are qualitatively different, since in the Markovian model, correlations decay exponentially, whereas at finite $\la$, the speed of decay is in general not faster than the decay of the correlation functions $f_{i,j}(t)$.}, see also \cite{duemcke}.  
A drawback of our technique with respect to the algebraic approach is that, in the case where $\initialres $ is a Gibbs state, it is not immediately clear that the invariant state $\rho_{\sys}^{inv}$ is the restriction of the coupled Gibbs state to the small system. 
However, if one extends the class of initial states as suggested above, this does immediately follow. 

\subsubsection{Comparison with earlier results}

One should distinguish between the case where the Gaussian reference state of the field has a non-zero density (temperature) in the thermodynamic limit, i.e.\ $\lim_{\La \nearrow \bbZ^d}   \str\La \str^{-1} \sum_{q \in \La^*}\eta^{\La}(q) >0$, or not. In the latter case, the field is essentially in the vacuum state and the approach to a steady state is related to the question whether the ground state of the coupled system (assuming that  it exists) is the only bound state and whether the rest of the spectrum is absolutely continuous. 
These questions have been extensively studied in  \cite{hubnerspohnspectral, bachfrohlichsigalqed, bachfrohlichsigalsoffer, goleniapaulifierz}.  In one sense, our results are sharper than those quoted: they cover cases where the coupled system has no ground state, yet there is approach to a steady state for the small system.  We do not explain nor develop these issues further here, but rather postpone them to a subsequent paper.  However, the quoted results  are  stronger in the sense that they allow for the confined system to have continuous spectrum above a ionization threshold.

If the field has a positive density, the prime example is of course the case where the field is in a thermal state at non-zero temperature, then the only results that we are aware of, rely on complex deformations of Liouvillians. One either uses complex translations or dilations. 
To streamline the discussion, we note that one can  rewrite the correlation functions $f_{i,j}(t)$ as 
\beq
f_{i,j}(t) = \int_{\bbR} \d \om  \,  \e^{\i t \om} \,   \hat f_{i,j}(\om), \qquad      \hat f_{i,j}(\om):=  \langle \overline{\varphi_{i, eff}(\om)},   \varphi_{j,eff}(\om)  \rangle_{\scrS}  \label{eq: correlation as fourier}
\eeq
such that $\varphi_{i,eff}$, the effective form factors (in the thermodynamic limit), are functions from $\bbR$ to some Hilbert space $\scrS $ that emerge naturally if one follows the operator-theoretic approach to the problem. They are often called "effective form  factors"(effective because they incorporate the density function $\zeta$ of the reservoir).  In the physical literature on the subject, the function  $\hat f_{i,j}(\om)$ is often called the 'spectral function'.

The first result on RTE, due to \cite{jaksicpillet2,jaksicpillet4}, proceeds by assuming that 
\begin{itemize}
\item The function $\om \mapsto \hat f_{i,j}(\om)$ is analytic in a strip of width $\ga_0$ such that $\om \mapsto \hat f_{i,j}(\om+\i\ga)$ is in $L^1(\bbR, \d\om)$  for $0 <\ga<\ga_0 $. 
\end{itemize}

 This of course corresponds to exponential decay of $f_{i,j}(t)$.  This result has been improved in \cite{derezinskijaksicreturn, derezinskijaksicspectral} where analyticity is replaced by demanding that $ \hat f_{i,j}(\cdot)$ is in $C^2$, implying $f_{i,j}(t) \sim \str t\str^{-2}$.  A related approach is found in \cite{froehlichmerkli}.

The approach via dilation analyticity has been pioneered by \cite{bachfrohlichreturn}. There one assumes that 
\begin{itemize}
\item the function $\om \mapsto \hat f_{i,j}( \e^{\i (\mathrm{sign} \om) \ga }\om)$ is in $L^1(\bbR, \d\om)$  for $0 <\ga<\ga_0 $  (this is dilation analyticity)
\item $\hat f_{i,j}(\e^{\i (\mathrm{sign} \om) \ga}\om)  \leq  \str\om\str^{1+\al}$ for some $\al>0$.  
\end{itemize}
By deforming the integration contour $ \bbR$  in \eqref{eq: correlation as fourier} into $\e^{-\i \ga}\bbR_- \cup \e^{\i \ga}\bbR_+ $ , one realizes that this implies that 
\beq
\str f_{i,j}(t) \str \leq   const \, t^{-(2+\al)} (\log t)^{const'}   
\eeq
and hence this case is covered by our result.

\subsection{Strategy of the proof} \label{sec: strategy of proof}

Our proof is based on a polymer expansion in real time. In the context of classical stochastic dynamics, such expansions were successfully applied in e.g.\ \cite{bricmontkupiainenexponentialdecay, maesnetocnyspacetime}, and in the case of classical deterministic dynamics in \cite{bricmontkupiainencoupledmaps}. For the case at hand, a similar strategy was pursued in \cite{deroeck}. In the following Sections \ref{sec: markov and leading},  \ref{sec: discussion polymer} and \ref{sec: dyson discussion}, we introduce the rough ideas. 
 
\subsubsection{Markovian approximation and leading dynamics} \label{sec: markov and leading} 

We discretize time $t=N\nu$ where $\nu$ is a macroscopic time unit $\nu = \la^{-2}\ell$, with $\ell$ a $\la$-independent number that could actually be chosen $\ell=1$.  Then,  we 
 write $\rho_{\sys,t}=Z_N \rho_{\sys,0}$ where $\rho_{\sys,t}$ is the reduced time-evolved density matrix and $Z_N$ could be called the 'reduced evolution operator'. The idea is that $T \equiv Z_{N=1}$ can be analyzed quite well, at least for sufficiently small coupling $\la$, because in that regime the Markovian approximation (Section \ref{sec: markov approximation}) can be justified.   Indeed, we will state in Section \ref{sec: dissipativity} that $T$ is well-approximated by $\e^{\la^2\nu L}$, with $L$ the Lindblad generator (also mentioned in Section \ref{sec: markov approximation}). This is not proven in the present paper since the proof is well-known in the literature.  For now, we view $T$ as the leading dynamics.  
 An important consequence of the fact that $T$ is close to $\e^{\la^2\nu L}$ and of Assumption \ref{ass: fermi golden rule}, is that we can establish that the operator $T$ has a  simple eigenvalue $1$ (this eigenvector is the 'steady state' $\rho_\sys^{T}$) and the rest of the spectrum lies in a circle  with radius $1-g <1$. Since $T$ is trace conserving, $\Tr T \rho_{\sys,0}= \Tr \rho_{\sys,0}$, the  'right' eigenvector corresponding to the eigenvector $1$ is the identity $1 \in \scrB(\scrH_\sys)$, hence we have the spectral decomposition
 \beq \label{eq: decomposition of T}
 T = R +(1-R) T, \qquad   R =  \str\rho_\sys^{T} \rangle \langle 1 \str, \qquad   \norm  (1-R) T^n \norm \leq  C (1-g)^n
 \eeq    
 This property of $T$ is proven in Section \ref{sec: dissipativity} by simple perturbation theory (with $\e^{\la^2\nu L}$ being the 'unperturbed object') , but it is introduced already in Section \ref{sec: general} as an assumption. 
 
 \subsubsection{Polymer representation} \label{sec: discussion polymer}
 
 If the reduced dynamics $Z_N$ were exactly Markovian, we would have $Z_N = T^N$, i.e.\ $Z_N$ could be called a 'quantum Markov chain'. 
 However, this is of course not the case and as $N$ grows the difference between $Z_N$ and $T^N$ becomes important.  We represent the corrections to $T^N$ by 'nonmarkovian excitations' that are localized in time.  For example, 
 \beq
 Z_2 = T^2 +  \caT[\bbE^c(B({1,2}))]
 \eeq
 where $ \bbE^c(B({1,2}))$ is an operator on $\scrB_1(\scrH_\sys) \otimes \scrB_1(\scrH_\sys)$ that should be thought of as localized in the macroscopic times $1$ and $2$ (actually, in macroscopic time intervals $[0, \nu]$ and $[\nu, 2\nu]$). The operation $\caT[\cdot]$ is a  time-ordering; it converts  $ \bbE^c(B({1,2}))$ to an operator on $\scrB_1(\scrH_\sys)$, such that it is on the same footing as $T$ (see the full definition in Section \ref{sec: correlation functions}. We are actually abusing the correct definition slightly in the present section). For $Z_3$, we get
 \beq
 Z_3 = T^3 + T \caT[\bbE^c(B({1,2}))] +   \caT[\bbE^c(B({2,3}))] T + \caT[\bbE^c(B(\{1,2,3\})]  +  \caT \left[ \bbE^c(B({1,3})) T(2) \right]
 \eeq
 where,  for a general set of macroscopic times $A$,  $\bbE^c(B(A))$  denotes the (irreducible) excitation that is localized in the elements of $A$ (it acts on the $\str A \str$-fold tensor power of $\scrB_1(\scrH_\sys)$). 
Since, in the rightmost term,  the excitation $\bbE^c(B({1,3})) $ is localized in times $1$ and $3$, and the $T$-operator represents the leading dynamics in the second time interval (hence the `$2$' in $T(2)$), we need to squeeze $T(2)$ in between the excitations at times $1$ and $3$. 

 For general $N$, the resulting expression for $Z_N$ is 
 \beq
Z_N   =   T^N+  \sum_{\caA  \in \textrm{Pol}(N)}    \caT   \left[ \left( \mathop{\bigotimes}\limits_{\tau \in I_N  \setminus \supp \caA} T(\dt) \right)   \bigotimes  \left(\mathop{\bigotimes}\limits_{A \in \caA} \bbE^c(B(A) ) \right)    \right]   \label{eq: Z from connected correlation functions preview}  
\eeq
where the \emph{polymer set} $\textrm{Pol}(N) $ is the set of nonempty collections $\caA$ of disjoint subsets  $A$ of $I_N= \{1,2, \ldots, N\}$.
To analyze this polymer expression, we use two tools: bounds on the excitation operators $\bbE^c( B(A) ) $ and a Feynman rule. \\
\textbf{Bounds} \, We will bound each term in the sum  \eqref{eq: Z from connected correlation functions preview}  in operator norm by  $ \prod_{A \in \caA} \norm \bbE^c( B(A) ) \norm_{\weird}$ where the norm $\norm \cdot \norm_{\weird}$ is defined in Section \ref{sec: norms}. The $T(\tau)$-operators do not show up in these bounds since they have norm $1$. 
 We will require that  $\norm \bbE^c( B(A) ) \norm_{\weird}  \sim \ep^{\str A \str}$ for some small parameter $\ep$ and, moreover, that $\norm \bbE^c( B(A) ) \norm_{\weird}$ decreases as the macroscopic times, i.e.\ the element of $A$; are further apart. This decrease as a function of temporal distance is a consequence of 
Assumption \ref{ass: decay of correlations}, but in Section \ref{sec: general} it is introduced as Assumption \ref{ass: integrability}. In Section \ref{sec: dyson}, we prove how Assumption \ref{ass: decay of correlations} implies Assumption \ref{ass: integrability}. 
 \\
\textbf{Feynman rule} \,   It is not hard to see that the bounds given above, when summed over the different terms in \eqref{eq: Z from connected correlation functions preview} lead to a too pessimistic bound on $Z_N$.  Even if we restrict to sets $A$ whose elements are consecutive integers (which is essentially justified because of the temporal decay), then we still get an exponentially diverging bound, of order $\e^{C \ep N}$, for some constant $C>0$. 
To improve our bounds, we use a Feynman rule (one could also call it a Ward identity) that is a consequence of conservation of probability of the dynamics $Z_N$, to be explained in Section \ref{sec: conservation of probability}.  In our general polymer expansion, this Feynman rule  implies that, for  every uninterrupted string of $T(\cdot)$ operators that follows a set $A$,  we can insert the spectral projection $(1-R)$ in front of the string of $T$'s.  By \eqref{eq: decomposition of T}, this
yields exponential decay in the length of the string. This is illustrated in Figure \ref{fig: connected1} (the sets $\hook(A)$ will be defined later).
\begin{figure}[h!]  
\psfrag{one}{ $1$}
\psfrag{last}{ $N$}
\psfrag{hookset1}{ $\hook(A_1)$}
\psfrag{hookset2}{ $\hook(A_2)$}
\psfrag{hookset3}{ $\hook(A_3)$}
\includegraphics[width = 16cm, height=3cm]{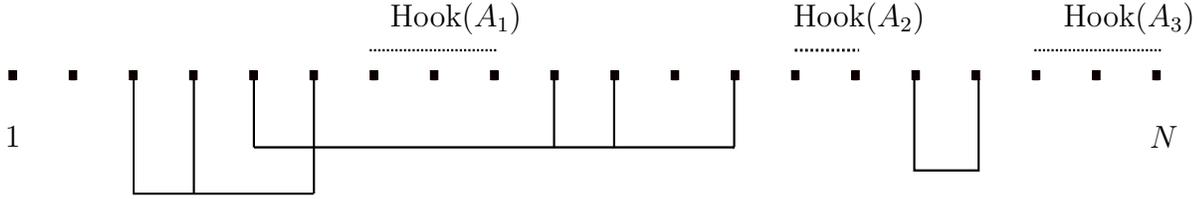}   
\caption{\footnotesize{An example of a $\caA$ with $\caA=\{ A_1,A_2, A_3\}$. In the picture, $N= 20$, and $A_1=\{ 3,4,6\}, A_2=\{ 5,10,11,13\}, A_3=\{16,17 \}$.    The exponential decay is on the string of times that are covered by the dotted lines. These strings are 
$  \hook(A_1)=\{ 7,8,9 \}, \hook(A_2)=\{ 14,15\}, \hook(A_3)=\{18,19,20 \}$.    These are exactly the times between $\max A$ for some $A$ and the next-in-time element of some other set $A'$.  }}  \label{fig: connected1}
\end{figure}

Armed with the Feynman rule and the bounds on $\bbE^c$, we can now perform the sum over all terms on the RHS of \eqref{eq: Z from connected correlation functions preview}, resulting in 
\beq
\norm Z_N - T^N \norm = O(\ep)  \quad   \Rightarrow  \quad  \norm Z_N -R \norm = C(1-g)^N + O(\ep)
\eeq 
By very similar arguments, one can see that $Z_N$, for large $N$ is $\ep$-close to a one-dimensional projector.  Indeed, by the exponential decay following any excitation, all contributions to $Z_N - T^N$ tend to be localized at times close to the final time $N$. This means that they start with a  string of $T$'s of length $O(N)$; such a string is equal to  $R$, up to an error of order  $(1-g)^{O(N)}$. Hence, up to a vanishing error,  all contributions get multiplied by $R$ on the right, and consequently they are of the form $ \str S \rangle \langle 1 \str$ for some  $S \in \scrB(\scrH_\sys)$. This means  that also the limit $\lim_{N \nearrow \infty}Z_N$ is of this form.  By conservation of trace and positivity it then follows that 
\beq
\lim_{N \nearrow \infty}Z_N =  \str \rho_\sys^{inv} \rangle \langle 1 \str, \qquad  \textrm{for some density matrix} \, \rho_\sys^{inv}:  \quad  \norm \rho_\sys^{inv}-\rho_\sys^{T} \norm = O(\ep), \qquad  \ep \to 0
\eeq
These arguments are presented in Section \ref{sec: proof of main theorem}. 

\subsubsection{Dyson expansion}  \label{sec: dyson discussion}
The Dyson expansion is introduced to prove the bounds on $\norm \bbE^c( B(A) ) \norm_{\weird}$ discussed above. This is done in Section \ref{sec: estimates}.   It is also the standard tool to prove the weak coupling limit, Theorem \ref{thm: weak coupling}.

As we will do later on in the proof, we assume  for simplicity that the set $\caI$ has just one element, such that we can drop the index $i \in \caI$ and simply write $H^{\La}_{\inter} = D \otimes \Phi(\varphi^{\La})$. For any operator $O $ on $\scrH^{\La}$,  let $O(t)= \e^{\i t (H_\sys + H^{\La}_\res)}  O  \e^{-\i t (H_\sys + H^{\La}_\res)}  $, and we abbreviate $ \Phi(t) \equiv (\Phi(\phi^{\La}))(t) $, then we can write the  Duhamel expansion  (the convergence of the series is easily established) 
\baq
\e^{\i t \adjoint(H_\sys)} \rho_{\sys, t} & :=&  \lim_{\La \nearrow \bbZ^d}  \e^{\i t H_\sys}  \Tr_{\res}    \left[ \e^{-\i t H^{\La}_\la} \left( \rho_{\sys, 0} \otimes \initialresfinite  \right)  \e^{\i t H^{\La}_\la}  \right]   \e^{-\i t H_\sys}     \label{eq: derivation of dyson}  \\
  &=&    \lim_{\La \nearrow \bbZ^d}  \sum_{n_{\links}, n_{\rechts} \in \bbN }  (-\i \la)^{n_\links}  (\i \la)^{n_\rechts}   \mathop{\int}\limits_{s_1 < \ldots < s_{n_\links}} \d s_1\ldots \d s_{n_\links}  
 \mathop{\int}\limits_{s'_1 < \ldots < s'_{n_\rechts}} \d s'_1\ldots \d s'_{n_\rechts}     \nonumber 
   \\[1mm]
  &&  \qquad   \Tr_{\res}    \left[ H^{\La}_{\inter}(s_{n_\links}) \ldots      H^{\La}_{\inter}(s_2)H^{\La}_{\inter}(s_1)   \left( \rho_{\sys, 0} \otimes \initialresfinite  \right)   H^{\La}_{\inter}(s'_{1})      H^{\La}_{\inter}(s'_2)\ldots  H^{\La}_{\inter}(s'_{n_\rechts})   \right]  
  \nonumber
   \\ [3mm]
&=&  \sum_{n_{\links}, n_{\rechts} \in \bbN }   (-\i \la)^{n_\links}  (\i \la)^{n_\rechts}        \lim_{\La \nearrow \bbZ^d}    \Tr_{\res}  \left[   \Phi(s_{n_\links}) \ldots \Phi(s_2) 
    \Phi(s_1) \initialresfinite  \Phi(s'_1)  \Phi(s'_2) \ldots   \Phi(s'_{n_\rechts})    \right]    \nonumber   \\[1mm]
&&\,
\mathop{\int}\limits_{0 <s_1 < \ldots < s_{n_\links}<t} \d  \underline{s}
  \mathop{\int}\limits_{0<s'_1 < \ldots < s'_{n_\rechts}<t}   \d  \underline{s}'   \, \,  
    D(s_{n_\links}) \ldots D(s_2)  D(s_1)  \rho_{\sys, 0} D(s'_1)  D(s'_2) \ldots   D(s'_{n_\rechts})    \nonumber 
\eaq

Note that the operators on the last line  act trivially on $\scrH^{\La}_\res$, and hence they are independent of the volume $\La$. 
The expression on the one-but-last line is a $n_\links+n_\rechts$-point correlation function corresponding to the Gaussian (quasi-free) state $\initialresfinite$.   Hence, by the Wick theorem, we can expand this correlation function into sums of pairings of the $n_\links+n_\rechts$ of products of two-point correlation functions $f^{\La}(\cdot)$, as in formula \eqref{eq: gaussian property}.  Each term in the sum is determined by a pairing of the $n_\links+n_\rechts$-times, it is called a "diagram". 

Starting from the Dyson expansion, we first identify which terms in that expansion make up the "excitation operators" $\bbE^c(B(A))$.  This is particularly intuitive. For example, the operator $\bbE^c(B(\{ \tau, \tau' \}))$ is built by all terms (diagrams) in the Dyson expansion whose times $\underline{s}, \underline{s}' $ fall into the domain 
$\nu[\tau-1, \tau] \cup \nu [\tau'-1, \tau']$ and such that at least one pair in the pairing connects the two intervals, i.e.\ it has one of its time-coordinates in each interval.  
For a general set of times $A$, the operator $\bbE^c(B(A))$ is made up by diagrams such that the set $A$ is connected by the pairs in that diagram 
 Two examples will be  given in Figure \ref{fig: diagrammicro}.

\section{Polymer model}  \label{sec: general}
In this section, we start from a discrete-time dynamical system and we derive the approach to a steady state, given some assumptions that will be justified in Section \ref{sec: discretization}. Apart from the first paragraphs, the discussion in this Section is independent of the setup given in Section \ref{sec: introduction}.  In particular, it could be applied without any change to other models, hinted at in points A and B of Section \ref{eq: restriction to confined}.

\subsection{Reduced dynamics and excitations} \label{sec: reduced dynamics}
 We define the propagator $U(\dt)$, implementing the dynamics between \emph{macroscopic} times $\tau-1$ and $\tau$, with $\dt \in \bbN$, and acting on joint density matrices $\rho_{\sys\res}$;
\beq
U(\dt)    :=   \e^{\i \nu \dt  \ad(H^{\La}_\res)}   \e^{-\i \nu  \ad(H_\lambda^{\La})}     \e^{-\i \nu (\dt-1)   \ad(H^{\La}_\res)} 
\eeq
where we have chosen the macroscopic times to be related to the microscopic times by  a  scaling factor $\nu$ that will be fixed in Section \ref{sec: discretization}, depending on details of the model.  The total dynamics (in the interaction picture) up to microscopic time $N$ is then given by 
\beq
U(N) \ldots    U(2)U(1) =   \e^{\i \nu N \ad(H^{\La}_\res)}   \e^{-\i \nu N  \ad(H^{\La}_\la)}  
\eeq
Furthermore, we define a projection operator $P$;
\beq
P \rho_{\sys\res} := \left(\Tr_{\res}   \rho_{\sys\res}\right)  \otimes \initialresfinite
\eeq
A distinguished role is played by the reduced dynamics $T(\dt)$, acting on density matrices $\rho_\sys$ on $\scrH_\sys$. It is defined by 
\beq
(T(\dt) \rho_\sys ) \otimes \initialresfinite  =   P U(\dt) (\rho_\sys \otimes \initialresfinite  )
\eeq
  Note that $T(\dt)$ is independent of $\dt$ by the stationarity property \eqref{eq: stationarity of reservoir}, but we still write the dependence on $\dt$ for bookkeeping reasons that will become clear in Section \ref{sec: correlation functions}.
It is convenient to abuse the notation and let $T(\dt)$ stand for $ T(\dt) \otimes 1$ as well, such that it acts on joint density matrices. (This abuse of notation appears only in the present section.) 
Note that, as operators on joint density matrices, 
\beq
T(\tau) (1-P) =  (1-P) T(\tau)  \neq 0
\eeq
Define the 'excitation operator' 
\beq
B(\dt) :=   U(\dt) - T(\dt) 
\eeq
and note that
\beq
P B(\dt) P =0
\eeq
Ultimately, we are interested in the reduced dynamics:
\baq
Z_N   &:=&  P  U(N) \ldots     U (2) U (1)  P =  P  \e^{-\i \nu N \ad(H^{\La}_\la) }  P \label{def: Z}
\eaq
We will next insert the decomposition $U(\dt)=T(\dt)+B(\dt)$.  We note that the definitions of $U(\dt), B(\dt)$ are given only  in  finite volume ($\str \La \str < \infty$). However, in Section \ref{sec: finite volume limits}, we will obtain some concepts (including $Z_N$) that do admit a thermodynamical limit.

\subsection{Correlation functions of excitations} \label{sec: correlation functions}

For notational purposes, it makes sense to define the following: Let $V$ be an operator in
$ \otimes_{\dt \in I_N}\scrR_{\dt}$  
where $I_N=\{ 1, \ldots,N \}$ and each $\scrR_\dt$ is a copy of $\scrR \equiv  \scrB_1(\scrH_\sys)$.
Define the 'time-ordering" $\caT$ as a linear operator $ \otimes_{\dt \in I_N} \scrB(\scrR_{\dt}) \to \scrB(\scrR): V \mapsto \caT [V]$ as follows. 
For elementary tensors $V= \otimes_{\tau \in I_N}V_\tau $ where  $V_\tau \in \scrR_{\tau}$, we simply put
\beq
\caT [V]  :=   V_N  \ldots V_2 V_1 
\eeq
and then we extend $\caT$ by linearity to the whole of $ \otimes_{\tau \in I_N} \scrB(\scrR_\tau)$. 
Now, take a subset $A \subset I_N$ and define the operator 
\beq
\bbE( B(A)):   \scrR_A \to \scrR_A, \qquad  \textrm{where} \, \,  \scrR_A :=    \otimes_{\dt \in A} \scrR_\dt
\eeq
as follows.  Let $A =\{ \tau_1, \ldots, \tau_m \}$, with the convention that $\tau_1 < \tau_2 < \ldots < \tau_m$, choose operators $S_\tau, S'_\tau$ in $\scrR_\tau, \scrR'_\tau$, respectively (here, $\scrR'_\tau$ is the dual space to $\scrR_\tau$) and let  $S_{A}=\otimes_{\tau \in A} S_{\dt} $ and  $S'_{A}= \otimes_{\tau \in A}  S'_{\dt} $   be elements of  $\scrR_A$ and $\scrR'_A$ , respectively . We define the operator   $ \bbE (B(A)) $ by giving its `matrix elements', namely
\baq
  \big\langle S'_A,    \bbE (B(A))  S_A \big\rangle^{} & := & \Tr \Big[ S'_{\dt_m}  B({\dt_m})       \left(\str S_{\dt_m}  \rangle    \langle  S'_{\dt_{m-1}} \str \otimes 1_\res \right)  \ldots \\[2mm]
&& \ldots      B({\dt_3})   \left(\str S_{\dt_3}  \rangle   \langle   S'_{\dt_2} \str \otimes 1_\res   \right)   B(\dt_2)    \left(\str S_{\dt_2}     \rangle \langle S'_{\dt_1}  \str \otimes 1_\res   \right)  B({\dt_1})       ( S_{\dt_1} \otimes \initialresfinite  ) \Big]
\eaq
where $\langle \cdot, \cdot \rangle$ on the LHS is the pairing between $\scrR'_A$ and $\scrR_A$ and, on the RHS,  the notation $\str A \rangle \langle A' \str$ for operators on $\scrR \sim \scrB_1(\scrH_\sys)$ was introduced following Assumption \ref{ass: fermi golden rule}.
Note   that $\bbE(B(A))=0$ whenever $A$ has only one element, since the operator $PB(\dt)P$ vanishes on tensor products of the form $S \otimes \initialresfinite $. 
Furthermore, the  correlation functions corresponding to sets $A$ and $A+\tau$ are copies of each other, but acting  on different spaces ($\scrR_A$ vs.\  $\scrR_{A+\tau}$).  This  follows from the stationarity of the reference states $\initialresfinite$ under the free reservoir dynamics. 
By expanding $U(\dt)=T(\dt)+B(\dt)$ for every $\dt$ in the expression for the reduced dynamics \eqref{def: Z}, we arrive at
\beq
Z_N =       \sum_{A  \subset  I_N }   \caT   \left[ \left( \mathop{\bigotimes}\limits_{ \dt \in I_N \setminus A} T(\dt) \right)   \bigotimes  \bbE(B(A) )     \right]  \label{eq: Z from correlation functions} 
\eeq
On the RHS, the first tensor  acts on $\scrR_{I_N \setminus A}$ and the second on $\scrR_{A}$. The time-ordering  $\caT$ makes sure the operators are 'contracted' in the right way.   Note that the order in which we write the tensors in expressions like \eqref{eq: Z from correlation functions}  does not have any significance.  Instead, the 'legs' of the tensor product on which the operators act are indicated by the arguments $\tau$ and $A$. 
\subsubsection{Connected correlation functions}

The "connected correlation functions", denoted by $\bbE^c(B(A))$, are defined to be operators on $\scrR_A$ satisfying
\beq
  \bbE(B(A')) =   
   \sum_{\scriptsize{\left.\begin{array}{c}  \caA \, \textrm{partitions of}\, A'    \end{array} \right.  }}  \left( \mathop{\bigotimes}\limits_{A \in \caA}  \bbE^c(B(A))  \right)   \label{def: cumulants}
\eeq
The tensor product in this formula makes sense since $ \scrR_{A'} = \otimes_{A \in \caA} \scrR_A$ whenever $\caA$ is a partition of $A'$.
Note that this definition of connected correlation functions reduces to the usual probabilistic definition when all operators that appear are numbers and the tensor product can be replaced by multiplication. 
Just as in the probabilistic case, the relations \eqref{def: cumulants} for all sets $A'$ fix the operators $\bbE^c(B(A))$ uniquely since the formula \eqref{def: cumulants} can be inverted. 

  With this machinery in place, we can write a neat expression for the reduced dynamics $Z_N$; 
\baq
Z_N      
&=& T^N +    \sum_{\caA  \in \textrm{Pol}(N)}   Z_N(\caA), \qquad \textrm{with}   \quad    Z_N(\caA):=       \caT   \left[ \left( \mathop{\bigotimes}\limits_{\tau \in I_N  \setminus \supp \caA} T(\dt) \right)   \bigotimes  \left(\mathop{\bigotimes}\limits_{A \in \caA} \bbE^c(B(A) ) \right)    \right]   \label{eq: Z from connected correlation functions}  
\eaq
where the \emph{polymer set} $\textrm{Pol}(N) $ is the set of non-empty collections $\caA$ of disjoint subsets  $A$ of $I_N$, and $\supp \caA= \cup_{A \in \caA} A $. 
Formula \eqref{eq: Z from connected correlation functions}   follows from \eqref{eq: Z from correlation functions} by substituting \eqref{def: cumulants}
 since, obviously, any polymer $\caA$ is a partition of $\supp \caA$. The term $T^N$ in \eqref{eq: Z from connected correlation functions}  originates from $A=\emptyset$ in \eqref{eq: Z from correlation functions}.

.

\subsubsection{Norms on $\scrB(\scrR_A)$} \label{sec: norms}
We  introduce a norm on the spaces $\scrB(\scrR_{A})$ that is appropriate for our bounds.  
Any operator $E $ on $\scrR_A$  can be written (usually in a non-unique way) as a finite sum of elementary tensors
\beq
E = \sum_\nu     E_\nu, 
\eeq
We  define
\beq \label{def: weird norm}
\norm E \norm_{\weird} := \inf_{\{ E_\nu \}}  \sum_{\nu}  \norm E_\nu \norm  
\eeq
where the infimum ranges over all such elementary tensor-representations of $E$.
The sole purpose of the norm $\norm \cdot \norm_{\weird}$ resides in the following estimate: for any $\caA \in \poly(N)$  and any operators $C_{\dt}$ on $\scrR_{\dt}$, we have
\beq
\left \norm \caT\left[       \left( \mathop{\bigotimes}\limits_{A \in \caA} \bbE^c( B(A))  \right)  \bigotimes   \left(  \mathop{\bigotimes}\limits_{\dt \in I_N \setminus \supp\caA } C_{\dt}   \right) \right]  \right \norm  \leq    
\mathop{\prod}\limits_{A \in \caA}  \norm \bbE^c( B(A)) \norm_{\weird}        \mathop{\prod}\limits_{\dt \in I_N \setminus \supp\caA } \norm C_{\dt} \norm
\label{eq: bound w norm}
\eeq
This follows immediately from the definition \eqref{def: weird norm}.

\subsection{Unitarity} \label{sec: conservation of probability}
We now examine the consequences of the conservation of probability,  i.e.\ unitarity of the propagators $U(\dt)$. 
\beq
 \Tr B (\dt)   \rho_{\sys\res}  =0, \qquad  \textrm{for any}  \qquad  \rho_{\sys\res} \in \scrB_1(\scrH_\sys \otimes \scrH_\res^{\La})
\eeq
This follows from the fact that 
\beq
 \Tr B (\dt)  \rho_{\sys\res} =    \Tr U (\dt)  \rho_{\sys\res}  -  \Tr (T(\dt)  \otimes 1) \rho_{\sys\res} =   \Tr U (\dt)  \rho_{\sys\res}  -  \Tr_\sys \Tr_\res U(\dt) \rho_{\sys\res}  =0
\eeq
Let $R$ be a one-dimensional projector on $\scrR$ of the form $R= \str \rho_\sys \rangle \langle 1 \str$, for some density matrix $\rho_\sys \in \scrB_1(\scrH_\sys)$ (i.e.  satisfying $\Tr \rho_\sys=1$). 
It follows that 
\beq
(1_{\scrB(\scrR_{A \setminus \max A})} \otimes R(\max A))  \,   \bbE( B(A)  ) =0, \qquad    (1_{\scrB(\scrR_{A \setminus \max A})} \otimes R(\max A))  \,   \bbE^c(  B(A )) =0   \label{eq: conservation of prob}
\eeq
where $R(\dt)$ is a copy of $R$ acting on $\scrR_\dt$ (here used for $\tau =\max A$)

\subsection{Infinite-volume setup} \label{sec: finite volume limits}

As remarked already, the above setup makes sense in finite volume $\La$ only. However, the operators $Z_N, T$ and the "correlation functions" $\bbE(B (A)),  \bbE^c(B(A))$ have well-defined thermodynamic limits, since these operators act on (tensor products of) the system space only.  More importantly, the relations 
\eqref{eq: Z from correlation functions}, \eqref{eq: Z from connected correlation functions} and \eqref{eq: conservation of prob} will remain valid in the thermodynamic limit.  
For the sake of explicitness, we put this in  a lemma
\begin{lemma} \label{lem: thermo limit of correlations}
The correlation functions $\bbE(B (A)),  \bbE^c(B(A))$, as introduced in Section \ref{sec: correlation functions}, depend on the volume $\La$ via the reference state $\initialresfinite$, and hence we should denote them by $\bbE^{\La}(B (A)),  \bbE^{\La,c}(B(A))$. However, the infinite-volume correlation functions
\beq
\bbE(B (A)):= \mathop{\lim}\limits_{\La \nearrow \bbZ^d} \bbE^{\La}(B (A)), \qquad  \bbE^{c}(B(A)) :=   \mathop{\lim}\limits_{\La \nearrow \bbZ^d}  \bbE^{c, \La}(B(A))
\eeq
exist and satisfy the equality \eqref{eq: conservation of prob}.
\end{lemma}
In the remainder of Section \ref{sec: general}, when writing $\bbE(B (A)),  \bbE^c(B(A))$ we always mean the infinite-volume quantities.  The proof of Lemma \ref{lem: thermo limit of correlations} is analogous to the proof of Lemma \ref{lem: thermodynamic limit}, which is contained in Section \ref{sec: combinatorics}. 

\subsection{Abstract result} \label{sec: abstract result}

We now state two assumptions that allow us to prove the convergence to a stationary state.
In Section 3 we verify these assumptions for the model considered in Section 1.

Our first assumption concerns the decay of truncated correlation functions. We demand that $k$-point truncated correlation functions have a sort of tree-graph decay reminiscent of high temperature Gibbs states. Let $A =(\dt_1,\ldots, \dt_k)$ with $\dt_i < \dt_{i+1}$ and put 
\beq
\dist_\zeta (A)= \dist_\zeta (\underline{\tau}):=   \prod_{i=1}^{k-1} \zeta(\str \dt_{i+1}-\dt_{i}\str) 
\eeq
with $\zeta$ a nondecreasing function satisfying \eqref{eq: subadditivity} and \eqref{eq: zeta subexponential}

\begin{assumption}[Summable correlations] \label{ass: integrability}
There is a  constant $\ep<\infty$  such that 
\beq
\sup_{\tau_0 \in I_N}\, \mathop{\sum}\limits_{\scriptsize{\left.\begin{array}{c}  A \subset I_N  \\   \max A=\tau_0   \end{array}\right.  }}  \ep^{-\str A \str}
\dist_\zeta (A)    \norm \bbE^c( B(A) )  \norm_{\weird}   \leq   1
\eeq
uniformly in $N$,
\end{assumption}

The next assumption expresses that the operator $T$ has a well-defined leading order contribution. 
\begin{assumption}[Dissipativity] \label{ass: dissipativity}
The operator $T$ has a simple eigenvalue $1$ corresponding to the one-dimensional spectral projector $R=\str \rho_\sys^T \rangle \langle 1 \str$ with $\rho_\sys^T$ a density matrix. Moreover, the rest of the spectrum lies inside a disk with radius $1-g$ with $g>0$. \end{assumption}
This assumption obviously implies that, for some constant $C_g<\infty$, 
\beq
\norm T^n(1-R) \norm  \leq C_g (1-g)^n, \qquad  \textrm{for any}\,  n \in \bbN
\eeq

Our main theorem states that under these assumptions, the system state approaches a unique limit as $N \to \infty$\bet \label{thm: main discrete}
If Assumptions \ref{ass: integrability} and  \ref{ass: dissipativity} hold with $\ep$ sufficiently small, then
there exists a density matrix $\rho_\sys^{inv}$ such that 
\beq
\lim_{N\to\infty}  Z_N =     \str \rho_\sys^{inv} \rangle\langle 1   \str  
\eeq
where 
\beq
 \left\norm  \rho_\sys^{T}-  \rho_\sys^{inv}   \right \norm ={\cal O}(\ep), \qquad   \ep \to 0
\eeq
The speed  of convergence is estimated as
\beq
 \left\norm  Z_N   -    \str \rho_\sys^{inv} \rangle\langle 1   \str  \right \norm  \leq  C_g  (1-g)^N + {\cal O}(\ep)\frac{1}{\zeta(N)},  \qquad   \ep \to 0   \label{eq: bound on decay to inv}
\eeq
\eet

\subsection{Proof of Theorem \ref{thm: main discrete}}  \label{sec: proof of main theorem}

\subsubsection{Summation of the polymer series}

We start from the representation
\beq
Z_N  =       T^N+   \sum_{\caA  \in \textrm{Pol}(N)}   Z_N(\caA)
\eeq
and we use the conservation of probability (Section \ref{sec: conservation of probability}) with 
\beq
R =  \str \rho^{T}_\sys \rangle\langle 1   \str, 
\eeq
where $\rho^{T}_\sys $ is the unique invariant state of the map $T$.   
It follows that 
\beq
Z_N(\caA)  =    \caT\left[       \left( \mathop{\bigotimes}\limits_{A \in \caA} \bbE^c( B(A))  \right)  \bigotimes   \left(  \mathop{\bigotimes}\limits_{\dt \in  \hook(\caA)} T(\dt)  (1-R)     \right)  \bigotimes \left(   \mathop{\bigotimes}\limits_{I_N \setminus (\supp \caA  \cup  \hook(\caA) )}  T(\dt)    \right)\right]    \label{eq: Z_N with hooks}
\eeq
where the set of times $\hook(\caA)$ is determined as follows. For each set $A \in \caA$, pick the smallest time in $\supp\caA \setminus A$ that is larger than $\max A$ and call this time $\dt_{\hook}(A)$.  If there is no such time (which happens for exactly one $A$,  namely the one for which $\max A= \max \supp \caA$),  then set $\dt_{\hook}(A)=N+1$.   We also define the sets 
\beq
\hook(A) := \left\{\max A+1, \max A+2, \ldots, \dt_{\hook}(A)-1 \right \}, \qquad \qquad    \hook(\caA) := \bigcup_{A \in \caA} \hook(A) 
\eeq
(in case  $\max A+1= \dt_{\hook}(A) $, we set $\hook(A):= \emptyset$.)
In words, the set $\hook(A)$ is simply the set of times between the latest time of $A$ and the next element that belongs to some $A' \in \caA$. See Figure \ref{fig: connected1} for a graphical representation of a term $Z_N(\caA)$. 

At this point, there is no evident relation between operators $Z_N(\caA)$ for different $N$.  However, let $m(\caA)=\min\supp(\caA)$ denote the earliest time included in the polymer $\caA$ and take a $\caA \in \poly(N') $ with $m(\caA) > N'-N$, for $N'>N$. Then, the polymer $\caA=\{ A_1, \ldots, A_m\}$ can be converted into a polymer in $\poly(N) $, 
namely  $ \caA -(N'-N) :=   \{A_1-(N'-N), \ldots, A_m-(N'-N) \}$ where  $A_i-(N'-N):= \{ \tau- (N'-N), \tau  \in A_i \} $.
By putting the projector $R$ on the right of $Z_N(\caA)$, we obtain then that
\beq
Z_{N'}( \caA) R =   Z_{N}( \caA -(N'-N) \})  R,  \qquad \textrm{if}  \, \,  m(\caA) > N'-N     \label{eq: relation between different N}
\eeq
This is a straightforward consequence of the translation invariance of the correlation functions.
Hence the difference between $Z_N R$ and $Z_{N'} R$ is made by polymers that are ``too extended" to be fitted into $I_N$ ($m(\caA)$ is too small); we have 
\baq
 &&(Z_{N'}- Z_N) - (T^{N'}-T^N)      \label{eq: crude difference}
 \\[1mm]
&& \qquad =      \mathop{\sum}\limits_{\scriptsize{\left.\begin{array}{c}  \caA \in \poly(N')  \\    m(\caA) \leq N'-N   \end{array}\right.  }} 
  Z_{N'}(\caA) R + 
 \mathop{\sum}\limits_{\scriptsize{\left.\begin{array}{c}  \caA \in \poly(N')    \end{array}\right.  }}  Z_{N'}(\caA) (1-R)-   \mathop{\sum}\limits_{\scriptsize{\left.\begin{array}{c}  \caA \in \poly(N)    \end{array}\right.  }}  Z_{N}(\caA) (1-R)  \nonumber     \eaq

\begin{lemma} \label{lem: bound on delta}
If  $\ep$  is small enough, then
\beq    \label{eq: bound on delta}
 \sum_{\caA  \in \textrm{Pol}(N)} \norm Z_N(\caA)  \norm  ={\cal O}(\ep), \qquad   \ep \to 0
\eeq
and 
\beq     \label{eq: delta excited to zero}
\lim_{N \to \infty} \sum_{\caA  \in \poly(N)} \norm Z_N(\caA)(1-R)  \norm = 0 
\eeq
\end{lemma}
\begin{proof}

First, we state a crucial bound that follows immediately form \eqref{eq: Z_N with hooks} by \eqref{eq: bound w norm} and the fact that $\norm T^n (1-R)  \norm \leq C_g(1-g)^n $ and $\norm T \norm =1$;  namely
\beq
\norm Z_N(\caA)  \norm \leq    \prod_{A \in \caA}   \norm \bbE^c( B(A))\norm_{\weird}  \, C_g (1-g)^{\str \hook(A)  \str} 
\eeq
Put 
\beq
 \delta_N   :=   \sum_{\caA  \in \textrm{Pol}(N)}    \prod_{A \in \caA}   \norm \bbE^c( B(A))\norm_{\weird}  \, C_g  (1-g)^{\str \hook(A)  \str}  \label{def: delta}
\eeq
Note that this sum over $\caA  \in \textrm{Pol}(N)$  can be read as the sum over all $\caA$ such that exactly one  of the $A \in \caA$ has $\dt_{\hook}(A)=N+1$, and for all other $A$, $\dt_{\hook}(A) \leq N$. Moreover, the next largest $\dt_{\hook}(A'), A' \in \caA$ necessarily belongs to $A$. We call this  next largest hook time $a \equiv \dt_{\hook}(A')$.   It follows that \beq \label{eq: induction on Z_N}
\delta_N  \leq   \sum_{A\subset I_N}    \norm \bbE^c ( B(A)) \norm_\weird   \, C_g (1-g)^{N-\max (A)}  \prod_{a \in A}  (1+\delta_{a})  \eeq
where we used that $\str \hook(A)  \str = N-\max A$.  
Next, observe that  $\delta_{N}$ is non-decreasing in $N$, hence one can replace $\delta_a$ on the RHS of \eqref{eq: induction on Z_N} by $\delta_{N-1}$. Using $g>0$ to control the sum over $\max(A)$ and assuming that $(1+\delta_{N-1})\ep\leq 1$, we get 
\beq
\delta_N \leq   \frac{ C_g }{g} 
\sup_{\tau_0 \in I_N}\, \mathop{\sum}\limits_{\scriptsize{\left.\begin{array}{c}  A \subset I_N  \\   \max A=\tau_0   \end{array}\right.  }}  (1+\delta_{N-1})^{\str A \str}
  \norm \bbE^c( B(A) )  \norm_{\weird} \leq
  \frac{\ep C_g }{g}  (1+\delta_{N-1}) 
\eeq
where in the last inequality Assumption \ref{ass: integrability} was used.  
Since $\delta_2 \leq \ep^2$ by inspection, we get by induction in $N$ that $\delta_N\leq  \frac{2\ep C_g }{g}$
for $\ep$ small enough. Hence eq.  \eqref{eq: bound on delta} holds.

To show \eqref{eq: delta excited to zero}, we first remark that 
\beq
\textrm{LHS of  \eqref{eq: delta excited to zero} } \leq \mathop{\sum}\limits_{\caA \in \poly(N)}    \norm Z_N(\caA) \norm C_g  (1-g)^{m(\caA)-1} 
\eeq
since we get $T(1-R)$  on the first $m(\caA)-1$ factors in the time ordered expression for $Z_N$.  Next, we set
\beq
v_{\leq}(N,n) := \mathop{\sum}\limits_{\scriptsize{\left.\begin{array}{c}  \caA \in \poly(N) \\ m( \caA) \leq n   \end{array}\right.  }}    \norm Z_N(\caA) \norm,   \qquad  v_{>}(N,n) := \mathop{\sum}\limits_{\scriptsize{\left.\begin{array}{c}  \caA \in \poly(N) \\ m( \caA) > n   \end{array}\right.  }}    \norm Z_N(\caA) \norm
\eeq
The relations between polymers for different $N$ implies that $v_{>}(N,n)= v_{>}(N+k,n+k)$ for any $k \in \bbN$. Moreover, the summability  \eqref{eq: bound on delta} implies that there is a  $\overline{v} < \infty$ such that
\beq
\overline{v} = \lim_{N \to \infty}  v(N), \qquad    v(N) :=  v_{\leq}(N,n) + v_{>}(N,n) \qquad  (\textrm{independently of $n$}) 
\eeq
Given $\kappa>0$, we choose  $N(\ka)$ such that $\str v(N(\kappa))-\overline{v}   \str \leq \kappa$ and hence $v_{\leq}(N+k,k) < \kappa$ for any $N \geq N(\kappa)$ and any $k \in \bbN$. 
Hence
\baq
\sum_{\caA \in \poly(N+k)}    \norm Z_N(\caA) \norm C_g(1-g)^{m(\caA)}   &\leq  & v_{>}(N+k, k))  C_g(1-g)^{k} +  v_{\leq}(N+k, k))     \\
   &\leq  &  \overline{v}\,  C_g (1-g)^{k} +  \kappa  
\eaq
As $k \to \infty$, this bound equals $\kappa$. Since $\kappa$ is arbitrary, this proves \eqref{eq: delta excited to zero}.

\end{proof}

\subsubsection{Convergence towards the steady state}

From the bounds in \eqref{eq: crude difference} and  Lemma \ref{lem: bound on delta}, we obtain
\beq
\limsup_{N' \to \infty}  \norm Z_N - Z_{N'} \norm   \leq   \limsup_{N' \to \infty}  \mathop{\sum}\limits_{\scriptsize{\left.\begin{array}{c}  \caA \in \poly(N') \\ m( \caA) \leq N'-N   \end{array}\right.  }}    \norm Z_{N'} (\caA) R \norm +      \mathop{\sum}\limits_{\scriptsize{\left.\begin{array}{c}  \caA \in \poly(N)    \end{array}\right.  }}   \norm Z_{N} (\caA)\norm   C_g  (1-g)^{ m( \caA)}   +  C_g (1-g)^N
\label{eq: bound on diff}
\eeq
  We will now estimate the two first terms on the RHS of \eqref{eq: bound on diff} multiplied by the factor $\zeta(N)$.  Note first  that for any  $\caA \in \poly(N')$ with $m( \caA) \leq N'-N$, 
  \beq
  N < N'-m(\caA) +1 = \sum_{A \in \caA}   \str \max A - \min A \str +   \sum_{A \in \caA}    \str \hook(A) \str,
  \eeq
  and hence, by property \eqref{eq: subadditivity} of the function $\zeta(\cdot)$, 
\beq  \label{eq: splitting the zeta}
\zeta(N) \leq \prod_{A \in \caA}   \dist_\zeta(A) \times  \zeta(\str \hook(A) \str).
\eeq
Hence 
\baq
\zeta(N)  \mathop{\sum}\limits_{\scriptsize{\left.\begin{array}{c}  \caA \in \poly(N') \\ m( \caA) \leq N'-N   \end{array}\right.  }}    \norm Z_{N'} (\caA) \norm     &\leq &  \mathop{\sum}\limits_{\scriptsize{\left.\begin{array}{c}  \caA \in \poly(N') \end{array}\right.  }}  \prod_{A \in \caA}       \dist_\zeta(A) \norm E^c(B(A) \norm_{\weird}  \zeta(\str \hook(A) \str)    \, C_g  (1-g)^{\str \hook(A)  \str}  \nonumber \\
&\leq&    \mathop{\sum}\limits_{\scriptsize{\left.\begin{array}{c}  \caA \in \poly(N')    \end{array}\right.  }}   \prod_{A \in \caA} c(\zeta,g)^{\str A \str}  \dist_{\zeta}(A) \norm \bbE^c ( B(A)) \norm_{\weird}   \,  \sqrt{1-g}^{\str \hook(A)  \str} 
\label{eq: clever bound diff}
 \eaq
where we have put (using that $\zeta$ is subexponential, see \eqref{eq: zeta subexponential})
 \beq  \label{def: c zeta}
c(\zeta,g):=\sup_{n \geq 1 }  \,  \zeta(n) ( \sqrt{1-g})^{n}  C_g   \qquad   \eeq
Note that we dropped the restriction that $m( \caA) \leq N'-N$ since it was only necessary for \eqref{eq: splitting the zeta}. 

Consequently, we have derived a bound, \eqref{eq: clever bound diff}, for the first term on the RHS of \eqref{eq: bound on diff} (since $\norm Z_N(\caA) R\norm  \leq \norm Z_N(\caA) \norm $).  We will now derive a similar bound  for the second  term on the RHS of  \eqref{eq: bound on diff}.

Instead of  \eqref{eq: splitting the zeta}, we use here that 
\beq  \label{eq: splitting the zeta 2}
\zeta(N) \leq     \zeta(m( \caA)) \times \prod_{A \in \caA}   \dist_\zeta(A) \times  \zeta(\str \hook(A) \str) , \qquad \textrm{for}\,  \caA \in \poly(N)
\eeq
and  we obtain the same bound as in \eqref{eq: clever bound diff}, except that this time we get $c(\zeta,g)^{\str A \str+1}$ instead of $c(\zeta,g)^{\str A \str}$ because of the presence of the term $\zeta(m(\caA))$ in \eqref{eq: splitting the zeta 2}.

Next, we show that \eqref{eq: clever bound diff} (or the analogous bound for the second term on the RHS of \eqref{eq: bound on diff}) can be bounded by $O(\ep)$, for $\ep$ small enough. 
To achieve this, we proceed in exactly the same way as in the proof of Lemma \ref{lem: bound on delta}, except that here;
\begin{itemize}
\item We include the factor $\dist_\zeta(A) $ in the weight  $\norm \bbE^c(B(\dt_{A})) \norm$, which is permitted by Assumption \ref{ass: integrability}.
\item For each set $A$, there is an additional factor $c(\zeta,g)^{\str A \str} $ (or $c(\zeta,g)^{\str A \str+1}$), which can be handled by choosing $\ep$ smaller.
\item  The factor $(1-g)$ is replaced by $\sqrt{1-g}$, which again forces $\ep$ to be smaller. 
\end{itemize}

To conclude the proof of Theorem \ref{thm: main discrete}, it remains to show that $Z_{\infty} := \lim_{N \nearrow \infty} Z_N$ is of the form $\str \rho_\sys^{inv} \rangle\langle 1   \str $. This follows from the fact that, by the bound  \eqref{eq: delta excited to zero}, only the terms $Z_N(\caA)R$ contribute to $Z_{\infty}$.
  \qed

\section{Discretization of the physical system}   \label{sec: discretization}

We explain now how the setup of Section \ref{sec: introduction} fits into the framework of Section \ref{sec: general}.   In Section \ref{sec: expansions}, we introduce and estimate the Dyson expansion and we present the construction of the operators $\bbE^c(B(A))$ from the microscopic model.  In Sections \ref{sec: treegraph} and \ref{sec: dissipativity}, we check Assumptions \ref{ass: integrability} and \ref{ass: dissipativity} starting from Assumptions \ref{ass: decay of correlations} and \ref{ass: fermi golden rule}.

\subsection{Expansions} \label{sec: expansions}

To save notation, we present the  case where there is only one element in the sum \eqref{def: interaction ham}  defining the interaction Hamiltonian, i.e.\ $\str \caI \str=1$ and we can write $D_i=D, \varphi_i=\varphi$. 
The general case can be treated in essentially the same way, we indicate the changes at the end of Section \ref{sec: combinatorics} and in the proof of Lemma \ref{lem: a priori}.

We expand the reduced dynamics, introduced in Section \ref{sec: general}; 
\baq
Z_N \rho_\sys &=& P U(N) U(N-1) \ldots U(1) P  \rho_\sys  = \Tr_{\res} \left[ (   \e^{-i \nu N H^{\La}_\la }    (\rho_\sys \otimes \initialresfinite)   e^{i  \nu N  H^{\La}_\la }   ) \right] 
\eaq
 in a Dyson series.

\subsubsection{Dyson expansion} \label{sec: dyson}

We let 
\beq
D ( t)  :=\e^{\i t \adjoint(H_\sys)}   D \e^{-\i t \adjoint(H_\sys)}
\eeq
and we recall the \emph{left} and \emph{right multiplication operators} $M_{\links}(A), M_{\rechts}(A)$ introduced in \eqref{eq: multiplication operators}.
Define the operator products 
\beq
K(\ut, \underline{k})   =   M_{k_{2n}}(\i D ( t_{2n})) \ldots M_{k_{1}}(\i D( t_1)), \qquad   \norm K(\ut, \underline{k})  \norm \leq \norm D \norm^{2n}     \label{eq: bound on K}
\eeq
with $\ut=(t_1,\ldots, t_{2n})$ is an ordered sequence of times $0< t_1 < \ldots <t_{2n}<t$ and $\underline{k} =(k_1,\ldots, k_{2n})$ is a sequence in $\{\links, \rechts \}$.  Next, recall the correlation functions $f^{\La}(t)$ (the labels $i,j$ have been dropped because of the simplification $\str \caI\str=1$)
and define 
\beq
G^{\La}(\ut, \underline{k})    :=  \sum_{pairings \, \pi}\prod_{(s,r) \in \pi}  \left\{ \begin{array}{ccc}    \la^2   f^{\La}(t_r-t_s)        &\textrm{if} &    k_{s}=\links \\[2mm]
   \la^2   \overline{f^{\La}(t_r-t_s) }    &  \textrm{if} &      k_{s}=\rechts  
 \end{array} \right.
\eeq
(recall the convention that $s<r$ in the pairing $\pi$). Finally,  we introduce the free $\sys$-dynamics
\beq
W_{\dt} := \e^{- \i \nu \dt  \adjoint (H_\sys)}
\eeq
 Then, we are ready to state the Dyson expansion for $Z_N$;
  \beq
Z_N =    W_{N} \sum_{n \in \bbN}   \mathop{\sum}\limits_{\underline{k} }  \mathop{\int}\limits_{0< t_1 < \ldots < t_{2n} < N \nu}  \d \ut  \,  K(\ut, \underline{k})   G^{\La}(\ut, \underline{k})  \label{eq: first dyson expansion}
\eeq
where the term with $n=0$ is defined to be $1$. 
This expression can be checked easily by expanding the evolution operator $\e^{-\i t \adjoint (H^{\La}_\la)}$ in powers of $\la$ and  writing the expectation values of the field operators in two-point contributions by the Wick theorem (this gives rise to the factor
 $G^{\La}(\cdot)$).    Even though the 'perturbation' $H^{\La}_{\inter}$ is unbounded, it is straightforward to check that the RHS equals the LHS, provided that the sum and integral on the RHS of \eqref{eq: first dyson expansion} converge absolutely, as will be derived in Section \ref{sec: combinatorics}.
 
 Some explicit intermediary steps of the derivation of \eqref{eq: first dyson expansion} were given in Section \ref{sec: dyson discussion}.  Here we arranged the ordered sets of times $\underline{s}, \underline{s}' $ (cfr.\ \eqref{eq: derivation of dyson})  into one ordered set of times $\ut$ and we used the $k$-labels to keep track of whether a time appeared in the expansion to the left or to the right of the density matrix.


\subsubsection{A formalism for the combinatorics} \label{sec: combinatorics}

The integral over ordered $\ut$, together with the sum over  \emph{left-right specifiers} $\underline{k}$
and
pairings, $\pi$,  on the set of times, is represented as an integral/sum over ordered pairs $ ((u_i,k_i^{u}),  (v_i, k_i^v))$ with $u_i,v_i \in \bbR^+$ and $k_i^{u}, k_i^{v} \in \{\links, \rechts\}$ and $i=1, \ldots, n$,  such that 
\beq
u_i < v_i, \qquad    u_1< \ldots < u_n
\eeq
This is done as follows. For any pair $(r,s) \in \pi$, we let $u_i=t_r, v_i =t_s$ and $k_i^{u}=k_{r}, k_i^{v}=k_{s}$ where the index $i=1,\ldots,n$ is chosen such that the $u_i$ are ordered: $u_1 < u_2 \ldots < u_n$. 
We represent one pair $ ((u_i,k_i^{u}),  (v_i, k_i^v))$ by the symbol $w_i$ and the $n$-tuple of them by $\uw$. 
We call $\Om_J$, with $J \in \bbR^+$, the set of  $\uw$ such that  $u_i,v_i \in J$  (for arbitrary $n$), and we use the shorthand
\beq
\mathop{\int}\limits_{\Om_J} \d \underline{w}  : = \mathop{ \sum}\limits_{n \geq 0} \,  \mathop{ \sum}\limits_{k_i^{u}, k_i^{v}} \,\mathop{ \int}\limits_{J^n} \d \underline{u} \mathop{ \int}\limits_{J^n} \d \underline{v}    \, \,  \chi[u_i <v_i] \chi[u_1 < \ldots < u_n]
\eeq
These new coordinates are illustrated in Figure \ref{fig: micro0}, where we relate them to the $\underline{s}, \underline{s}' $-coordinates that were used in Section \ref{sec: dyson discussion}.
\begin{figure}[h!]  
\psfrag{nul}{ $0$}
\psfrag{time}{$t$}
\psfrag{le1}{ $s_1$}
\psfrag{le2}{ $s_2$}
\psfrag{le3}{ $s_3$}
\psfrag{le4}{ $s_4$}
\psfrag{ri1}{ $s'_1$}
\psfrag{ri2}{ $s'_2$}
\psfrag{ri3}{ $s'_3$}
\psfrag{ri4}{ $s'_4$}
\psfrag{ri5}{ $s'_5$}
\psfrag{ri6}{ $s'_6$}
\includegraphics[width = 15cm, height=5cm]{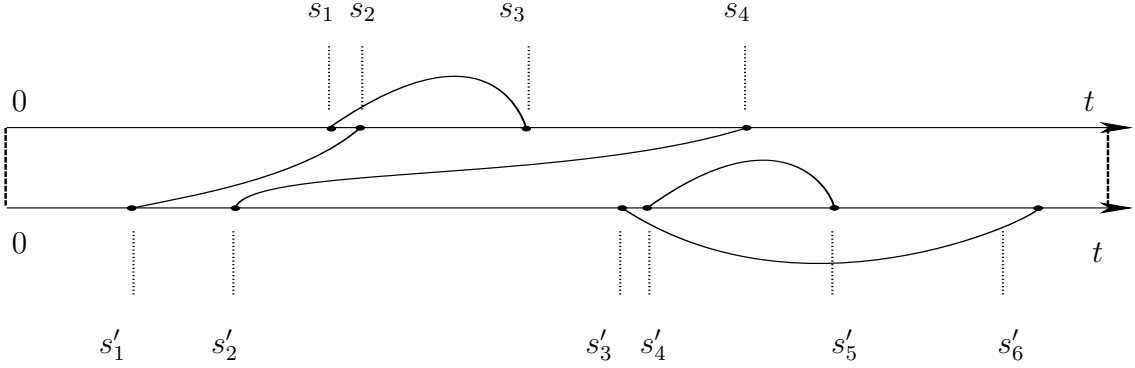}   
\caption{\footnotesize{Example of a term ("a diagram") in the Dyson expansion.  Each bent line carries a scalar factor $f(\Delta s)$ or $\overline{ f(\Delta s)}$ where $\Delta s$ is the difference between the times at both ends of  that line.  The dots at the end of the line carry operator valued factors $D(s), D(s')$.  In the example, $n_\links=4$ and $n_\rechts=6$. We  order the pairs according to their first time, which is called $u_i$  and the corresponding later time is called $v_i$.  The index $i$ is fixed by requiring $u_i \leq u_{i+1}$. With these $(u_i,v_i)$-coordinates, the above diagram would correspond to 
$
(u_1,v_1)=(s'_1 , s_2), (u_2,v_2)=(s'_2 , s_4)    , (u_3,v_3)=(s_1 , s_3), (u_4,v_4)=(s'_3 , s'_6)   ,  (u_5,v_5)=(s'_4 , s'_5)     
$
The left/right labels $k^{u}_i, k^{v}_i $ would then be $k^{v}_1, k^{v}_2, k^{u}_3,k^{v}_3 = \links$ and $k^{u}_1,k^{u}_2,k^{u}_4,k^{u}_5, k^{v}_4,k^{v}_5= \rechts$. That is, whenever a time is on the upper line, it has the label $\links$ and, whenever it is on the lower line, it has the label $\rechts$. 
 }}  \label{fig: micro0}
\end{figure}

In this new notation,  we can write a simple term-by-term bound on the Dyson expansion, using the bound in \eqref{eq: bound on K} 
\beq
\sum_{n \in \bbN}   \mathop{\sum}\limits_{\underline{k} }  \mathop{\int}\limits_{0< t_1 < \ldots < t_{2n} < t}  \d \ut  \,  \norm K(\ut, \underline{k}) \norm  \left \str  G^{\La}(\ut, \underline{k}) \right\str \leq \mathop{\int}\limits_{\Om_{[0,t]}}   \d \uw \,  \,\prod_{i=1}^n  \, \la^{2}\,   (\norm  D\norm)^{2} \,  \str f^{\La}(v_i-u_i) \str \label{eq: bound on K in w} \eeq
 The advantage of the last formula is that the sum over pairings is now represented by the integrals over $u_i, v_i$.   
One can first perform the integrals over $v_i$, bounding each of them by  $\la^2 \norm D \norm^2  \int_0^t \d s \str f^{\La}(s)\str $, and then the integral over $u_i$ gives $t^{n}/n!$, i.e.\  the volume of the $n$-dimensional simplex.   Hence  \eqref{eq: bound on K in w} is bounded by
\beq
\exp\left(  4\la^2 \norm D \norm^2   t  \int_0^t \str f^{\La}(s)\str  \d s \right)   
 \label{eq: dysonbound}
 \eeq
where the factor $4$ comes from the sum over $k_i^{u},k_i^{v} $. 
Since the functions $f^{\La}_{i,j}(\cdot)$ converge to $f_{i,j}(\cdot)$, uniformly on compacts,  the Dyson expansion converges absolutely and uniformly in the volume $\Lambda$;
\beq
\sum_{n \in \bbN}   \mathop{\sum}\limits_{\underline{k} }  \mathop{\int}\limits_{0< t_1 < \ldots < t_{2n} < t}  \d \ut  \,  \norm K(\ut, \underline{k}) \norm  \left \str  G^{\La}(\ut, \underline{k}) \right\str \leq   \exp\left( \la^2 t  ( \norm h \norm_1+ o(\str\La\str^0))\right) , \qquad  \textrm{as} \, \La \nearrow \bbZ^d   \label{eq: overall bound finite}
\eeq
where the positive function $h (t) $ was defined in \eqref{def: function h} for a general index set $\caI$.
It  follows that (here we indicate explicitly the $\La$-dependence of $Z_N$)
\beq
\lim_{\La \nearrow \bbZ^d} Z^{\La}_N  = W_N   \sum_n  \mathop{\int}\limits_{0 < t_1 <\ldots < t_{2n} < \nu N } \d \ut  \, \sum_{\underline{k}}    \,  K(\ut,\underline{k})    G(\ut, \underline{k}) 
\eeq
where $ G(\ut, \underline{k}) = \lim_{\La \nearrow \bbZ^d}  G^\La(\ut, \underline{k}) $ is obtained by simply replacing $f^\La(\cdot)$ by $f(\cdot)$. 
We have now proven Lemma \ref{lem: thermodynamic limit}, up to the simplification that $\str \caI \str=1$. However, one can trivially check that in the case $\str \caI \str >1$, the bound \eqref{eq: overall bound finite} remains true. 
Finally, we introduce the correlation function 
\beq
G(\underline{w})    :=  \prod_{((u,k^{u}),(v,k^{v})) }  \left\{ \begin{array}{ccc}    \la^2   f(v-u)        &\textrm{if} &    k^{u}=\links \\[2mm]
    \la^2  \overline{f(v-u) }    &  \textrm{if} &      k^{u}=\rechts  
 \end{array} \right.  \label{eq: correlation function w}
\eeq
where the product runs over the pairs $((u,k^{u}),(v,k^{v}))$ in $\uw$. 

\subsubsection{Connected correlations and the Dyson series} \label{sec: treegraph}

To relate the previous sections to the setup in Section \ref{sec: general}, we need to discretize time and express the operators $\bbE^c(B(A))$ in terms of the Dyson series.

 Recall that $I_N= \{ 1, \ldots, N \}$  is the set of macroscopic times.  To a set $A \subset I_N$ of macroscopic times, we associate the  domain of microscopic times
 \beq
 \Dom_\nu (A)= \bigcup_{\tau \in A}  [\nu(\tau-1), \nu\tau] 
 \eeq
 A set of pairs $\uw \in \Om_{ \Dom_\nu (A)}$ determines a graph $\caG_A(\uw)$ on $A$ by the following prescription: the vertices $\dt< \dt'$ are connected by an edge if and only if there is pair $w=((u,k^{u}),(v,k^{v}))$ in $\uw$ such that 
\beq
u \in  \Dom_\nu ( \dt) =   [ \nu (\dt-1),  \nu \dt ] \qquad  \textrm{and} \qquad     v \in  \Dom_\nu ( \dt') =  [ \nu (\dt'-1),  \nu \dt' ]
\eeq 
We write $\supp(\caG_A(\uw))$ for the set of non-isolated vertices of $\caG_A(\uw)$, i.e.\ the vertices that have at least one connection to another vertex.  (this set  $\supp(\caG_A(\uw))$ is obviously a subset of $A$).

Now, our task is to connect the Dyson  expansion in sets of pairs with the discrete-time expansion that was introduced in Section \ref{sec: general}.
The connection is via 
\baq
 && \left( \mathop{\bigotimes}\limits_{\dt \in A} W_{-\dt} \right)  \bbE(B(A))   \left( \mathop{\bigotimes}\limits_{\dt \in A} W_{\dt-1} \right)  \\[3mm]
&&  \qquad = \mathop{\int}\limits_{\scriptsize{\left.\begin{array}{c}  \uw \in \Om_{\Dom_\nu (A)}   \\   \supp(\caG_A(\uw)) =A    \end{array}\right.  }}  \,   \d \uw   \,   \,    \caI_{\tau(t_{2n})}\left[ M_{k_{2n}}(\i D(t_{2n})) \right]  \ldots \caI_{\tau(t_1)}\left[ M_{k_1}(\i D(t_1)) \right]   \,   G(\uw) 
\eaq
where  $ \caI_{\tau}$ is the embedding that  maps operators on $\scrR$ into operators on $\scrR_\tau$, and $\tau(t)$ assigns to each microscopic time $t$ the right "interval" $\tau$ (that is $\tau =\tau(t) \leftrightarrow t \in  \Dom_\nu ( \dt)$).  The coordinates $t_i,k_i$ on the RHS are determined implicitly by the $\uw$ as explained in Section \ref{sec: combinatorics}.  The product of correlation functions $ G(\uw)  $ has been defined in \eqref{eq: correlation function w}. 

The above formula follows immediately from the definition of $B(\cdot)$. Indeed, sets of pairs $\uw$ such that $\supp(\caG_A(\uw))$ is strictly smaller than $A$, say $\tau \in A \setminus  \supp(\caG_A(\uw))$, can be resummed to give $T(\tau)$, and hence they do not contribute to $B(\tau)$.  Indeed; 

\beq
W_{-\tau}T(\dt)W_{\tau-1}= \mathop{\int}\limits_{\scriptsize{\left.\begin{array}{c}  \uw \in \Om_{\Dom_\nu(\dt)}   \end{array}\right.  }}  \,   \d \uw   \,   \,    \caI_{\tau(t_{2n})}\left[ M_{k_{2n}}(\i D(t_{2n})) \right]  \ldots \caI_{\tau(t_1)}\left[ M_{k_1}(\i D(t_1)) \right]   \, G(\uw)   \label{eq: T in terms of caI}
\eeq
where now all ${\tau}(t_j)=\tau$ and hence all embeddings $\caI_{{\tau}(t_j)}$ are into $\scrR_\tau$.

 Next, we state a useful formula for the "truncated correlation functions".
\begin{lemma} \label{lem: truncated is spanning} 
Assume the simplification $\str \caI\str=1$, then 
\baq
 &&  \left( \mathop{\bigotimes}\limits_{\dt \in A} W_{-\dt} \right)  \bbE^c(B(A))   \left( \mathop{\bigotimes}\limits_{\dt \in A} W_{\dt-1} \right) \nonumber
  \\[3mm]
 && \qquad  =  \mathop{\int}\limits_{\scriptsize{\left.\begin{array}{c}  \uw \in \Om_{ \Dom_\nu (A)}   \\  \caG_A(\uw) \, \textrm{connected}    \end{array}\right.  }}  \d \uw \, \,  \caI_{\tau(t_{2n})}\left[ M_{k_{2n}}(\i D(t_{2n})) \right]  \ldots \caI_{\tau(t_1)}\left[ M_{k_1}(\i D(t_1)) \right]   \, \, \,  G(\uw)   \eaq
\end{lemma}
See Figure \ref{fig: diagrammicro} for two examples of $\uw$ that contribute to  $ \bbE^c(B(A)) $ for a given $A$. 
\begin{figure}[h!]  
\psfrag{one}{ $1$}
\psfrag{nul}{ $0$}
\psfrag{last}{ $N$}
\psfrag{time}{ $t$}
\psfrag{t1}{ $\nu$}
\psfrag{t2}{ $2\nu$}
\begin{center}
\includegraphics[width = 10cm, height=6cm]{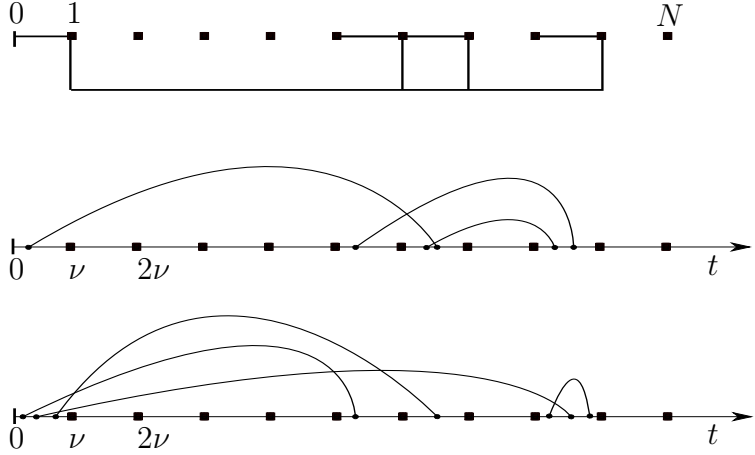}   
\caption{\footnotesize{We illustrate the discretization of the perturbation series.  The numbers $0,1,\ldots, N$ refer to macroscopic times measured in units of $\nu$.  The top drawing shows (by the graph under the horizontal axis)  the set $A=\{ 1,6,7,9 \}$. This set in fact refers to the time-intervals that end in these points, i.e.\ the intervals that are marked by horizontal lines.  In the two bottom pictures,  we have drawn two examples of diagrams (without indicating the $k$-labels) that induce a connected graph on $A$ and hence contribute to $\bbE^c(B(A))$. }}  \label{fig: diagrammicro}
\end{center}
\end{figure}

\begin{proof}

Let 
\beq
F(\uw) := K(\ut,\underline{k})    G(\uw)    \label{def: function F}
\eeq
where, again, the $\ut,\underline{k}$ are determined implicitly by $\uw$, as explained in Section \ref{sec: combinatorics}.
Obviously, we have the decomposition, 
\baq
&&  \mathop{\int}\limits_{\scriptsize{\left.\begin{array}{c}  \uw \in \Om_{ \Dom_\nu (A')}    \\   \supp(\caG_{A'}(\uw)) =A'    \end{array}\right.  }}   \d \uw \,  F(\uw)  \\
&& =   \sum_{\scriptsize{\left.\begin{array}{c}  \caA \, \textrm{partitions of}\, A'    \end{array} \right.  }}    \left[\prod_{A \in \caA} 
 \mathop{\int}\limits_{\scriptsize{\left.\begin{array}{c}  \uw_A \in \Om_{ \Dom_\nu (A)}   \\  \caG_A(\uw_A) \, \textrm{connected}    \end{array}\right.  }} 
 \d \uw_A  \right]      F(\mathop{\cup}\limits_{A \in \caA}  \uw_A)
\eaq
where  $\mathop{\cup}\limits_{A \in \caA}  \uw_A$ is well-defined since $\uw_A$ are sets (of  pairs of times). 
Since the connected correlation functions are uniquely defined by \eqref{def: cumulants}, the claim of the lemma follows.
\end{proof}

\subsection{Estimates} \label{sec: estimates}

In this Section, we verify Assumption \ref{ass: integrability} of Section \ref{sec: general}. 
The following estimate is an immediate consequence of Lemma  \ref{lem: truncated is spanning} and  the bounds at the end of Section \ref{sec: dyson}. 
\begin{lemma} \label{lem: a priori}

\beq  \label{eq: minimal spanning bound}
 \norm \bbE^c( B(A))  \norm_{\weird}   \leq     \e^{ \la^2 \norm h \norm_1   \str \Dom_\nu(A) \str  } \int_{\Om_{ \Dom_\nu (A)} }\d \uw  \, \,       \left( \prod_{i=1}^n  \la^2 h(v_i-u_i) \right) \,   \chi[ \underline{w} \, \textrm{spans} \, A\, \textrm{minimally}]
\eeq
where 
 the statement "$\uw$ spans $A$ minimally" means that $\caG_A(\uw)$ is connected and that no pair can be dropped from $\uw$ without losing this property. In particular, this implies that  
 $\caG_A(\uw)$ is a spanning tree on $A$. 
\end{lemma}

\begin{proof}
Let $F(\cdot)$ be as in \eqref{def: function F}, then 
\beq
 \mathop{\int}\limits_{\scriptsize{\left.\begin{array}{c}  \uw \in \Om_{ \Dom_\nu (A)}   \\  \caG_A(\uw) \, \textrm{connected}    \end{array}\right.  }}  \d \uw  \,  \norm F(\uw) \norm  \quad   \leq \quad   \mathop{\int}\limits_{\scriptsize{\left.\begin{array}{c}  \uw' \in \Om_{ \Dom_\nu (A)}   \\   \underline{w}' \, \textrm{spans} \, A\, \textrm{minimally}    \end{array}\right.  }}   \d \uw'   \quad     \mathop{\int}\limits_{  \uw'' \in \Om_{ \Dom_\nu (A)} }  \d \uw''   \norm F(\uw' \cup \uw'') \norm
\eeq
This appealing estimate was the main motivation for encoding the pairings $\pi$ in the elements $\uw$.


  To realize why it holds true,  choose a spanning tree $\scrT$ for the connected graph $\caG_A(\uw)$ and then  pick a minimal subset $\uw'$ of the pairs in $  \uw$ such that $\caG_A(\uw')=\scrT$.    
  Since, in general, this can be done in a nonunique way, the integrals on the RHS contain the same $\uw$ more than once, and the inequality is strict unless $F$ is concentrated on minimally spanning $\uw$.  
Starting from this inequality, we now use the bound \eqref{eq: overall bound finite} to perform the integral over $\uw''$, using also  
 that $ \norm F(\uw) \norm \leq   \prod_{i=1}^n  \la^2 h(v_i-u_i) $.
\end{proof}

\noindent {\bf Remark.}  This lemma holds in the general case $\str \caI \str \geq 1$. To see this
one can first extend the setup slightly, by making $\uw$ a set of pairs $(u,k^{u},a^{u} ; v,k^{v},a^{v} )$ where $a^{u},a^{v} $ range over the elements in $\caI$. Then, the reasoning goes through in exactly the same way and at the end the sum over $a^{u}, a^{v}$ can be performed (note indeed that the function $h$  includes a double sum over $\caI$, see \eqref{def: function h}) such that the bound \eqref{eq: minimal spanning bound} remains true with the original definition of $\uw$.

\vskip 2mm

We now prove that Lemma \ref{lem: a priori} implies Assumption \ref{ass: integrability}. Recall the decay time $\frt_L$ associated to the Markov approximation by Assumption \ref{ass: fermi golden rule}. 
\begin{proposition}\label{prop: integrability}
Set the scale factor  $\nu \equiv \la^{-2}\ell$ where  $\la$ is the coupling strength and $\ell \in [\frt_L,2\frt_L]$. Fix also a nondecreasing function $\zeta(\cdot)$ satisfying \eqref{eq: subadditivity} and such that $\int \d t \zeta(t) h(t) < \infty $. 
   Then, for $\la$ small enough there exists $\ep=\ep(\lambda)$ with  $\ep(\lambda)\to 0$ as $\str \la \str \to 0$
   such that
   \beq  
\sup_{\tau_0 \in I_N} \, \mathop{\sum}\limits_{\scriptsize{\left.\begin{array}{c}  A \subset I_N  \\    \max(A) =\tau_0   \end{array}\right.  }} \, \ep^{-\str A \str} 
\dist_\zeta (A)   \norm \bbE^c( B(A) )  \norm_{\weird}   \leq   1    \label{eq: bound to be proven}
\eeq
uniformly for $\ell \in [\frt_L, 2 \frt_L]$ and $N$. 
 \end{proposition}
\begin{proof} 
Each  $\uw$ that spans $A$ minimally determines a spanning  tree on $A$. Hence we can reorganise the bound \eqref{eq: minimal spanning bound} by first integrating all $\underline{w}$ that determine the same tree. This amounts to integrate, for each edge of the tree, all pairs $(u,v)$ that determine this edge. Hence we arrive at the bound
\beq  \label{eq: spanning trees bound}
 \norm \bbE^c( B(A))  \norm_{\weird}   \leq     \e^{ \la^2 \norm h \norm_1     \str \Dom_\nu(A) \str}     \,   \sum_{ span. trees \, \scrT  \, on \, A} 
\prod_{  \{ \dt,\dt' \} \in \caE(  \scrT) }   \hat e(\dt,\dt')   
\eeq
where $\caE(\scrT)$ is the set of edges of the tree $\scrT$,
\baq
 \hat{e}(\dt,\dt') &:=&  \int_{  \nu[\dt, \dt+1]} \d v  \int_{  \nu[\dt', \dt'+1]}  \d v \,  \,  \la^2  h(v-u), \qquad \textrm{if}\, \,   \tau < \tau'
 \eaq
 and   $\hat{e}(\dt',\dt):= \hat{e}(\dt,\dt')$ for   $\tau' < \tau$.
To bound $ \hat{e}(\dt,\dt') $, we  use 
\baq
 \hat{e}(\dt,\dt')  &\leq & \la^2 \nu   \mathop{\int}\limits_{  \nu(\dt'-\dt-1)}^{  \nu(\dt'-\dt+1)} \d s  h(s), \qquad  \textrm{for} \quad      \dt'-\dt>1  \\[1mm]
  \hat{e}(\dt,\dt+1)    & \leq&  \la^2 \int_0^{2\nu} \d s   \, s h(s) 
\eaq
Using the integrability assumption $\int \d t \zeta(t) h(t) < \infty$, the inequality $  \zeta(\str \dt-\dt' \str) \leq \zeta(\str \dt-\dt' \str-1)\zeta(1) $ from \eqref{eq: subadditivity}, and the fact that $\zeta$ is nondecreasing,   we easily derive  that 
\beq
\sum_{\dt' \in \bbN \setminus \{ \dt \}}     \zeta(\str \dt-\dt' \str)  \hat e (\dt,\dt')  \leq  \ep_1, \qquad  \textrm{with} \qquad \ep_1=\ep_1(\la) \to 0 \quad \textrm{as}  \quad \str \la \str \to 0    \label{eq: integral over edges}  \eeq
 uniformly in $\ell \in [\frt_L,2 \frt_L]$.   The key observation here is that  for any integrable function $f(\cdot)$, 
\beq
 \frac{1}{t}\int_0^t \d s  \,  s f(s)   \quad  \mathop{\longrightarrow}\limits_{t \to \infty}  \quad 0
\eeq
which follows by the dominated convergence theorem applied to the sequence of functions $F_t(s):= \chi[s <t] \frac{s}{t} f(s)$, since $\forall s: \lim_{t \nearrow \infty} F_t(s)=0 $.     
Note that our bound \eqref{eq: integral over edges} is not very ambitious in that we did not try to use the fact that $\zeta$ increases already  on the microscopic scale.

We now turn to the sum in \eqref{eq: spanning trees bound}. Recall the distance function$\dist_{\zeta}(A)$ introduced in Section \ref{sec: abstract result}. Since $\zeta$ is an increasing function, the inequality
\beq
\dist_{\zeta}(A) \leq \prod_{ \{ \dt,\dt' \} \in \caE(\scrT) } \zeta(\str \dt'-\dt \str)    \label{eq: inequality spanning trees}
\eeq
holds for   any spanning tree $\scrT$ on $A$. Here is a procedure for checking \eqref{eq: inequality spanning trees}: Let $\caE_{\mathrm{lin}}$ be the set of edges of the linear tree, i.e.\  $ \caE_{\mathrm{lin}} = \{\{\tau_i, \tau_{i+1} \}, i=1, \ldots, \str A \str-1 \} $ where $\tau_1, \ldots, \tau_{\str A \str}$ are the time-ordered elements of $A$. Then we need to prove 
\beq \label{eq: edges set}
  \prod_{e \in \caE_{\mathrm{lin}}}  d_e  \leq    \prod_{e \in \caE(\scrT)}  d_e   =  d(\caE(\scrT)),  \qquad  d_{\{ \tau,\tau' \}}  =  \zeta(\str \dt'-\dt \str) 
\eeq
Choose a leaf $\tau $ of the tree $\scrT$ and let $\tau'$ be the unique vertex that shares an edge (that we call $e_\tau$) with $\tau$. If $e_\tau \in \caE_{\mathrm{lin}}$, then \eqref{eq: edges set} is equivalent to the claim where we replace $\caE_{\mathrm{lin}}$ by $\caE_{\mathrm{lin}} \setminus e_\tau$ and $\scrT$ by the tree with the vertex $\tau$ and the edge $e_\tau$ removed.  If $e_\tau \notin \caE_{\mathrm{lin}}$ then consider the tree $\tilde \scrT$ where the edge $e_\tau$ is removed and the edge $\tilde e_\tau =\{ \tau, \tilde \tau\}$ is added, where $\tilde \tau$ is chosen such that  $\tilde \tau- \tau$ has the same sign as $\tau'-\tau$ and $\tilde e_\tau \in \caE_{\mathrm{lin}}$.  Clearly, $d(\caE(\tilde\scrT)) \leq d(\caE(\scrT))$ and hence it suffices to prove  \eqref{eq: edges set} with $\scrT$ replaced by $\tilde \scrT$.   We can continue by induction. \\

\noindent Using  that $ \la^2 \str\Dom_\nu(A)\str \leq 2\frt_L \str A \str$  in the exponent on the RHS of \eqref{eq: spanning trees bound},  the LHS of \eqref{eq: bound to be proven} is bounded by
\baq
 \sup_{\tau_0 \in \bbN } \mathop{\sum}\limits_{\scriptsize{\left.\begin{array}{c}  A \subset \bbN:   A \ni \tau_0, \str A \str \geq 2 \end{array}\right.  }}  \, \,   \ep^{-\str A \str} \e^{  const \str A\str  }  \,    \sum_{ span. trees \, \scrT  \, on \, A}  \, \,
\prod_{ \{ \dt,\dt' \} \in \caE(\scrT) }  \zeta( \str\dt'-\dt\str)  \hat e(\dt,\dt')   \label{eq: sum over trees}
\eaq
(Note that we also replaced  the constraint $\max A=\tau_0$ by $A \ni \tau_0 $).   The sum over $A \ni \tau_0$ and spanning trees on $A$ is of course equivalent to a sum over trees containing $\tau_0$. The control of such sums is at the core of cluster expansions, see e.g.\ \cite{ueltschi}.  
We state a simple result, Lemma \ref{lem: combinatorics}, proven in a more general setting in \cite{ueltschi}.  \end{proof}

\begin{lemma}\label{lem: combinatorics}
Let $\hat v(\cdot, \cdot)$ be a symmetric and positive function on $\bbN\times \bbN$ such that 
\beq
\sup_{\tau}  \sum_{\tau' \in \bbN \setminus \tau } \hat v(\tau, \tau') \    \leq  \frac{\ka}{2} \e^{-\ka}, \qquad  \textrm{for some} \,  \ka  >0 
\eeq
Then
\beq  \label{eq: combi sum to control}
\sup_{\tau_0}   \mathop{\sum}\limits_{\scriptsize{\left.\begin{array}{c} \textrm{trees} \,   \scrT \, \textrm{on} \, \bbN    \\  1 < \str  \scrT \str < \infty, \scrT \ni \tau_0   \end{array}\right.  }}   \prod_{(\tau, \tau') \in \caE(\scrT)}   \hat v(\tau, \tau')  \leq    \ka 
\eeq
\end{lemma}
We apply this lemma to our case with $\hat v(\tau, \tau') = \ep^{-1} \zeta( \str\dt'-\dt\str)  \hat e(\dt,\dt') $ and $\ka = \e^{const} \frac{\ep_1}{\ep}$ for $\frac{\ep_1}{\ep} $ sufficiently small.  It follows that \eqref{eq: sum over trees} is bounded by $ {\cal O}(\frac{\ep_1}{\ep^2})$.   This means that one can take $\ep \to 0$ as $\la \to 0$, since $\ep_1(\la) \to 0$.

\subsection{The dissipativity condition on $T$} \label{sec: dissipativity}

In this section, we check Assumption \ref{ass: dissipativity}. Observe first that the reduced dynamics $T=T_{\la,\ell}$ depends on both $\la$ and $\ell$ (via the scale factor $\nu$).   
We quote a celebrated result (originally in \cite{davies1} under slightly more stringent assumptions, see \cite{duemcke,derezinskideroeck2} for the result as quoted here)
\bet[weak coupling limit]  \label{thm: weak coupling}
Assume that Assumption \ref{ass: decay of correlations} holds. Then, for any $\ell_{max}<\infty$, we have

\beq
\lim_{\la\to 0}\sup_{\ell < \ell_{max}}   \left\norm  \e^{ \i \ell \la^{-2}  \adjoint ( H_{\sys})  }  T_{\la,\ell}  - \e^{ \ell    L  }  \right\norm=0
\eeq
where  $L$ is the generator given in Section \ref{sec: introduction}. 
\eet
\noindent\textbf{Remark} In fact, the proof of Theorem \ref{thm: weak coupling} is rather straightforward if one uses the framework of the present paper.  In \eqref{eq: T in terms of caI}, one shows that the contribution of those $\uw$ that contain two pairs $(u_i,v_i), (u_j,v_j), i<j$ such that $u_j \leq  v_i$ (i.e.\ the pairs are 'entangled') vanishes in the limit $\la \to 0$. This follows from the same simple calculation as done in the proof of Proposition \ref{prop: integrability} to show that $\sum_{\tau'} \hat e(\tau,\tau')$ vanishes as $\la \to 0$.  Then, all what remains is to prove that,  the 'ladder diagrams' $\uw$, i.e.\ those for which $v_i < u_{i+1}$, sum up to produce the semigroup $\e^{\la^2t L}$, in the interaction picture and up to an error that vanishes as $t=\la^{-2} \frt, \la \to 0$. \\ 

By Assumption \ref{ass: fermi golden rule}, we know that the semigroup $\e^{ \frt L}$ is exponentially ergodic (see \eqref{eq: decay scale}). Since  the free dynamics  $\tilde W_{\ell}:= \e^{ -\i \ell \la^{-2}  \adjoint ( H_{\sys})  } $ is an isometry on $\scrB_1(\scrH_\sys)$, we deduce from \eqref{eq: decay scale} that

\beq
\norm \tilde W_{\ell} \e^{\frt  L }- \tilde W_{\ell} \str \rho_\sys^{L} \rangle \langle 1 \str \,  \norm  \leq \e^{- \frt/ \frt_L}, \qquad  \textrm{for} \, \frt > \frt_L
\eeq
 Since $L$ commutes with $\adjoint (H_\sys)$, we also have  $ \tilde W_{\ell}  \rho_\sys^{L}  =  \rho_\sys^{L} $. 
Hence, to check Assumption   \ref{ass: dissipativity}, we merely need to choose $\la$ small enough  such that
\beq
 T_{\la,\ell} -   \tilde W_{\ell} \e^{\ell L} 
\eeq
is a sufficiently small perturbation of  $\tilde W_{\ell}\e^{\ell L} $. In that case,  spectral perturbation theory of isolated eigenvalues applies and we find  that  $T_{\la, \ell} $ has  a unique maximal eigenvalue (which, by unitarity, is necessarily equal to $1$)  corresponding to the eigenvector 
$ \str \rho_\sys^T(\ell)  \rangle \langle 1 \str$ where $ \rho_\sys^{L}-\rho_\sys^T(\ell) = o(\str\la\str^0)$, as $\la \to 0$.  The appearance of $\langle 1\str$ and the fact that $\rho_\sys^T(\ell) $ is  a density matrix,  are consequences of the fact that  the dynamics preserves positivity and the trace. Moreover, 
\beq
\left\norm \left(T_{\la, \ell} \right)^n  -   \str \rho_\sys^T(\ell)  \rangle \langle 1 \str  \right\norm \leq   (1-g(\ell))^n
\eeq
where 
\beq   \label{eq: form of g ell}
1-g(\ell) = \e^{-\ell/\frt_L} +o(\str\la\str^0)  =  \exp{\{ -  \frac{\ell}{\frt_L+o(\str\la\str^0) }  \} }, \qquad \la \to 0 
\eeq
In the above  estimates,  the error term  $o(\str\la\str^0)$ depends on $\ell$ as well; in general, a smaller $\la$ is needed when $\ell$ grows.  We restrict  $\ell \in [\frt_L, 2 \frt_L]$, which allows us to make all estimates uniform in $\ell$. 

The conclusion is that, for $\la$ small enough,  Assumption  \ref{ass: dissipativity}  holds for $T=T_{\la,\ell}$,  with $\ell \in [\frt_L, 2 \frt_L]$, 
$g(\ell)$ as in \eqref{eq: form of g ell} and $C_g=1$.

\section{Proof of Theorem  \ref{thm: main} and Proposition \ref{prop: speed of convergence}}

Finally, we weld the abstract results of Section \ref{sec: general} to the setup presented in Section \ref{sec: introduction}, using the estimates of Section \ref{sec: discretization}.

\subsection{Theorem  \ref{thm: main}}
First, we choose $\la$ small enough such that 
\begin{itemize}
\item  Assumption  \ref{ass: dissipativity} holds with $T=T_{\la,\ell}$  for $\ell \in [ \frt_L, 2\frt_L]$ and such that $g$ can be chosen uniformly in $\ell$. The possibility of doing this was sketched above in Section \ref{sec: dissipativity}. 

\item Proposition \ref{prop: integrability} applies with $\ep(\la)$ small enough such that  Assumption \ref{ass: integrability} holds and Theorem  \ref{thm: main discrete} applies. 
%
\end{itemize}
Let us abbreviate the reduced dynamics $V_t := P \e^{-\i t \adjoint (H_\la)}P$, then  Theorem  \ref{thm: main discrete} yields that
\beq \label{eq: convergence w}
V_{N \ell \la^{-2}}   \quad  \mathop{\longrightarrow}\limits_{N \to \infty}    \quad    \str \rho_{\sys}^{inv}(\ell) \rangle \langle 1 \str. 
\eeq
Moreover, we  deduce that $\rho_{\sys}^{inv}(\ell)$ is continuous in $\ell$.  This follows from 
the absolute summability of the series for $Z_{\infty}$ (by \eqref{eq: bound on delta}), uniformly in $\ell$, and the continuity of $ \bbE^c( B(A) )$ and $T $ as functions of $\ell$ (which follows easily from bounds on the Dyson expansion). 

Next, we argue  that $\rho_{\sys}^{inv}(\ell)$ does not depend on $\ell$. Indeed for $\ell_1, \ell_2$ such that $\ell_1/\ell_2$ is rational, we can pick a subsequence of  $(V_{N \ell_1 \la^{-2}})_N$ that is identical to a subsequence of  $(V_{N \ell_2 \la^{-2}})_N$. Hence the limits of both sequences coincide.  Since  $ \ell \rightarrow \rho_{\sys}^{inv}(\ell)$ is  continuous, it is necessarily constant. 
This ends the proof of Theorem  \ref{thm: main}.  \qed
\subsection{Proposition \ref{prop: speed of convergence}}

We start from \eqref{eq: bound on decay to inv} where we put $C_g=1$ (as argued in Section \ref{sec: dissipativity}, this is possible)
\beq
 \left\norm  V_{N \ell \la^{-2}}    -    \str \rho_\sys^{inv} \rangle\langle 1   \str  \right \norm  \leq   (1-g(\ell))^N + {\cal O}(\ep) \frac{1}{\zeta(N)},  \qquad   \ep \to 0   \label{eq: bound on decay to inv repeated}
\eeq
Since any $t >\la^{-2}\frt_L $ can be written as $t= \la^{-2} N \ell$ for some $N >0$ and $\ell \in [ \frt_L, 2\frt_L] $, we obtain 
\beq
 \left\norm  V_{t}    -   \str \rho_\sys^{inv} \rangle\langle 1   \str  \right \norm  \leq   \exp{\{ -  \frac{\la^2 t}{\frt_L+o(\str\la\str^0) }  \} }  + {\cal O}(\ep(\la))\frac{1}{\zeta(\la^2 t /\ell)},  \qquad   \la \to 0   
\eeq
where we used \eqref{eq: form of g ell}.  After using that $\ep \to 0$ as $\la \to 0$, and $\zeta(\la^2 t /\ell ) \geq \zeta(\la^2 t /(2 \frt_L) )$, we obtain Proposition \ref{prop: speed of convergence}.   \qed

\bibliographystyle{plain}
\bibliography{mylibrary08}

\end{document}